\documentclass[aps,prb,twocolumn,superscriptaddress,nofootinbib,longbibliography]{revtex4-2}
\DeclareMathAlphabet{\mathbbold}{U}{bbold}{m}{n}
\usepackage{amsmath,amssymb,amsthm}
\usepackage{graphicx}
\usepackage{booktabs}
\usepackage[hidelinks]{hyperref}

\newcommand{\Om}{\Omega}
\newcommand{\BZ}{\mathrm{BZ}}
\newcommand{\dd}{\mathrm{d}}
\newcommand{\kxy}{\kappa_{xy}}
\newcommand{\Var}{\operatorname{Var}}
\newcommand{\Cov}{\operatorname{Cov}}

\newcommand{\kB}{k_{B}}
\newcommand{\e}[1]{\times10^{#1}}

\newtheorem{theorem}{Theorem}
\newtheorem{lemma}[theorem]{Lemma}
\newtheorem{proposition}[theorem]{Proposition}

\begin{document}

\title{Extracting Chern numbers from the equilibrium transport noise of
multiband bosons}

\author{Zhi-Wei Wang}
\affiliation{College of Physics, Jilin University, Changchun 130012, China}
\author{Samuel L.\ Braunstein}
\affiliation{Department of Computer Science, University of York, York YO10 5GH, United Kingdom}

\date{\today}

\begin{abstract}
We show that the equilibrium transport noise of a multiband bosonic
system carries a measurable imprint of its Chern numbers, and we give
the protocol that extracts them. For electrons the topological part of
the Berry curvature is absent from the transport noise: over a filled
band it yields the quantized anomalous Hall conductivity, yet at a Fermi
surface the Chern number is invisible to the fluctuations of the Hall
current. That cancellation is an accident of the Fermi surface: magnons
and phonons populate every band thermally, each with its own Chern
number, so the harmonic sector of the Hodge decomposition is not a
single constant, and a variance does not annihilate it. What survives is
the occupation-weighted dispersion of Chern numbers. That current noise
sees band geometry is known, and at zero temperature the antisymmetric
part of the noise sum rule already returns a Chern number; the new
element is an object with no zero-temperature analogue, existing only
when several bands are populated and weighted differently. We derive it
from a fluctuating Boltzmann equation, place it in the standard
slow-variable framework where noise weights project onto conserved
quantities, and resolve the energy-magnetization subtraction at the
level of fluctuations: the subtracted term is a curl, contributing
exactly zero to the current a transport measurement records, realization
by realization, and no magnetization auto-correlation enters the
antisymmetric cross-spectrum at any wavevector.
Measurement uses an equilibrium sum
rule on that cross-spectrum, with no drive and hence no suppression: its
frequency-integrated form cannot in practice separate topological from
geometric content, for a rank reason, whereas the frequency-resolved
form does, at accessible resolution. The protocol is specified for
Cu(1,3-bdc), with calibration and robustness
characterized. A driven alternative lies twelve orders below the
equilibrium floor. Two independent implementations agree on every
convention-free quantity.
\end{abstract}

\maketitle

\tableofcontents

\section{Introduction}
\label{sec:intro}

The Berry curvature of a Bloch band performs several jobs. Integrated
over a filled band, it yields a Chern number and an anomalous Hall
conductivity that is quantized and insensitive to deformation. Weighted
by a Fermi surface, it yields a Berry-curvature dipole and a nonlinear
Hall response that is not quantized at all and is exquisitely sensitive
to band structure. Weighted by a Bose factor, it yields a thermal Hall
conductivity that is neither. These are usually treated as separate
calculations sharing a common ingredient. They are better understood as
separate \emph{projections} of one object, the projection being the
Hodge decomposition of the curvature two-form on the Brillouin zone.

This paper asks what that decomposition says about the second moment of
the current rather than the first. For electrons the answer is clean,
and usefully negative: the topological sector of the curvature is silent
in the noise. The proof is short and rests on two hypotheses. One is a
conservation law. The other, invisible while it holds, is that the
topological sector is a single number. At a Fermi surface it is, because
essentially one band matters. For bosons it is not, and this paper
follows out the consequences. Section~\ref{sec:overview} states the
observable that results, the protocol that measures it, and the material
prediction, ahead of the derivation.

We have written the paper so that a reader acquainted with the Berry
phase but not with differential forms, or familiar with fluctuating
hydrodynamics but not band topology, can follow it without consulting
anything else. Sections~\ref{sec:hodge} and~\ref{sec:visibility} are
accordingly elementary and may be skipped by anyone who already knows
that a closed two-form on a torus splits into a constant piece carrying
the Chern number and a remainder that integrates to zero.

\subsection{What we claim, and what we do not}
\label{sec:claims}

Since several of the ingredients are standard in one field or another,
we state up front which parts we regard as ours. The paper rests on
three layers of prior work.

The \emph{kinetic machinery} is not ours. The fluctuating Boltzmann equation is due to
Kogan and Shul'man~\cite{KoganShulman}, and the constrained equal-time correlator used in
Sec.~\ref{sec:langevin}, including the rank-one subtraction that enforces particle-number
conservation, appears explicitly in Gantsevich, Gurevich and Katilius~\cite{GGK}. Our
derivation in that section applies their framework; it does not establish it.

The \emph{spectral-weight decomposition} is likewise not ours. That the slow spectral
weights of an observable are its projections onto the conserved quantities is the standard
structure of fluctuating hydrodynamics: Mori's projection
formalism~\cite{Mori1965}, Kadanoff and Martin's derivation of correlation functions from
conservation laws~\cite{KadanoffMartin}, and, in its most familiar concrete form, the
Landau--Placzek ratio of the Rayleigh line to the Brillouin lines~\cite{Mountain}.
Section~\ref{sec:collision} is the Landau--Placzek decomposition applied to the Berry
curvature, and we present it as such.

The \emph{proposition that band geometry is visible in current noise} is also not ours,
and it is the closest prior work. Neupert, Chamon and Mudry showed, for a zero-temperature
band insulator, that the current-noise spectrum carries the quantum geometric tensor, that
its frequency integral obeys a sum rule whose antisymmetric part is the Hall conductivity
and hence returns the Chern number, and that for two bands the frequency-resolved noise
gives the geometric tensor integrated over the contour of constant direct
gap~\cite{NeupertChamonMudry}. Both structural moves used below therefore have a
zero-temperature electronic ancestor, and the sum rule itself belongs to a documented
family of frequency moments of the conductivity in which the zeroth Hall moment is
proportional to the Chern number. Those authors also state the obstruction: for more than
two bands the sum over bands cannot be performed, because the energies depend on the band
index too, so that energetics and quantum geometry combine. That obstruction is the
subject of Secs.~\ref{sec:parti} and~\ref{sec:partii}, and solving it at finite temperature
is what this paper adds to theirs. Two later results bound the problem on either side. In
time-reversal-invariant fermionic systems the driven noise carries an
occupation-variance-weighted square of the anomalous velocity, the fermionic analogue of
the driven object of Appendix~\ref{app:driven}~\cite{WeiEtAl}; there the weight is
unconstrained rather than conservation-constrained, and every Chern number vanishes by
symmetry, so the question asked here cannot be posed there. In the opposite limit of
dispersionless bands the frequency integral does return an integer~\cite{KruchkovRyu2024},
for reasons taken up in Sec.~\ref{sec:parti}.

One point of contact must be stated plainly, because the two statements sound contradictory
without being so. The electronic silence theorem below concerns the transport noise of
\emph{occupation fluctuations}, computed from a kinetic equation carrying a conservation
constraint, at finite temperature. The zero-temperature result just quoted concerns the
ground-state current-current correlator, whose content is interband particle-hole
coherence and which has no occupations to fluctuate. Both go by the name current noise;
they are different objects, and only the first is silent.

The electronic statement, that the topological sector is silent in the transport noise,
is the starting point of this paper rather than one of its results. It is established in
the framework paper~\cite{HodgeFramework}, and it is also recovered here as a corollary
rather than taken on trust: the kinetic derivation of Sec.~\ref{sec:langevin} applies
verbatim with the fermionic weight $\mathcal D=f(1-f)$, and at a Fermi surface the normalized weight
is supported on essentially one band, so the topological term of Eq.~\eqref{eq:varsplit}
vanishes identically, which is hypothesis~(H2) doing its work.

Our claims are the following. First, that the electronic cancellation has two hypotheses
rather than one, and that the second fails for bosons (Sec.~\ref{sec:failure}). Second,
that the surviving quantity is the occupation-weighted dispersion of Chern numbers, which
we have not found treated as an observable and which cannot arise at zero temperature,
where the occupations are zero or one and there is no dispersion across bands to weight
(Sec.~\ref{sec:failure}). Third, the resolution of the energy-magnetization subtraction at
the level of fluctuations: a vanishing theorem for lead-measured noise, and a protection
lemma showing that the magnetization auto-correlation cannot enter the antisymmetric
cross-spectrum at any wavevector (Sec.~\ref{sec:magnetization}).
Fourth, a conditioning result: for dispersive bands
the frequency-integrated equilibrium observable is identifiable in principle but not in
practice, the direction that moves one Chern unit between the upper bands lying some three
orders below the model-error floor, for a reason of rank (Sec.~\ref{sec:parti}).
Fifth, a protocol that
can, with its resolution requirement, experimental design and error budget
(Secs.~\ref{sec:partii}--\ref{sec:material}). The fourth and fifth together are the
finite-temperature multiband answer to the obstruction stated in
Ref.~\cite{NeupertChamonMudry}, which is the sense in which this paper continues that one.

\section{Observable, protocol and prediction}
\label{sec:overview}

The object this paper delivers is most usefully stated in advance. Write $C_\lambda$
for the Chern number of band $\lambda$ and $\bar n$ for the Bose occupation. The kinetic
theory of Secs.~\ref{sec:langevin} and~\ref{sec:failure} attaches to the transverse
transport noise a normalized thermal weight $p$, built from the occupation fluctuation
$\bar n(1+\bar n)$, and shows that the noise contains the \emph{variance} under $p$ of
the per-band harmonic constants of the Berry curvature, which are the Chern numbers up to
the fixed factor $2\pi/A_\BZ$, with $A_\BZ$ the Brillouin-zone area. At a Fermi surface
that variance is identically zero for electrons, because essentially one band is weighted
and a single band has a single Chern number; for bosons it is nonzero whenever two
populated bands carry different Chern numbers (Proposition~\ref{prop:main}). The
observable is thus the occupation-weighted dispersion of Chern numbers across populated
bands, an object without a zero-temperature analogue. For the kagome ferromagnet at the
material coupling $D/J=0.15$, with $D$ the Dzyaloshinskii--Moriya coupling and $J$ the
exchange, it is $0.38$, $1.84$ and $5.44$ per cent of the noise variance at $T=0.5$, $1$
and $3\,JS$, with $JS$ the magnon energy unit (Table~\ref{tab:shares}).

An equilibrium sum rule fixes where it sits in data (Sec.~\ref{sec:sumrule}). With
$S^{A}=\frac{1}{2i}(S_{xy}-S_{yx})$ the antisymmetric part of the symmetrized current
cross-spectrum between orthogonal components, the $\tanh$-weighted frequency integral of
$S^{A}$ returns the occupation-weighted first moment of the curvature, whose harmonic
content is a known linear form in the Chern numbers. No drive is applied, so none of the
gradient suppression that rules out the driven route (Appendix~\ref{app:driven}) is
incurred. The integral alone, however, cannot separate the topological from the geometric
content: for dispersive bands the temperature-and-field sweep of the integrated
observable is a rank-deficient transform (Sec.~\ref{sec:parti}). The protocol therefore
works with the integrand before integration, resolved in frequency and binned in the
band-pair gap coordinate, where the pair decomposition of Sec.~\ref{sec:partii} turns the
recovery of the Chern numbers into a discrete selection against a nuisance basis rich
enough to be honest.

The prediction is concrete. For Cu(1,3-bdc), the kagome ferromagnet whose magnon bands
and thermal Hall effect are established, neutron parameters give $J=0.6$~meV, $S=1/2$ and
$D/J\approx0.15$, so $JS=0.3$~meV, and a Zeeman gap of $0.3\,JS$ requires about $1$~T
(Sec.~\ref{sec:material}). The spectrum turns on at the minimal direct gap
$2\sqrt3\,D=0.52\,JS$, which is $38$~GHz, and the band-pair gap support extends to
$473$~GHz. The protocol asks for $S^{A}$ across that window in one hundred and forty-four
bins of about $3$~GHz, at nine settings, three temperatures by three fields, and recovers
the Chern triple $(1,0,-1)$, in the anchored convention of Appendix~\ref{app:sign}, by
discrete selection. At that resolution the model error is $9.0\e{-3}$ and recovery of the
triple at five per cent additive noise is one thousand of one thousand trials
(Fig.~\ref{fig:gap}); the selection survives Lorentzian broadening of the data out to a
width of $0.35\,JS$, two thirds of the turn-on quoted above
(Sec.~\ref{sec:discussion}), and the
calibration burden is met by co-estimating the band structure from the same spectrum
(Sec.~\ref{sec:robust}). The window carries a built-in null test: the intrinsic theory
predicts $S^{A}$ identically zero below the turn-on, so any antisymmetric weight measured
on $[0,2\sqrt3\,D)$ reads the systematic error directly.

The derivation fills the body of the paper. Sections~\ref{sec:hodge}
and~\ref{sec:visibility} set up the decomposition at tutorial level;
Secs.~\ref{sec:electronic}--\ref{sec:failure} establish the observable;
Secs.~\ref{sec:collision} and~\ref{sec:magnetization} give its spectral anatomy and settle
the energy-magnetization subtraction at the level of fluctuations; and
Secs.~\ref{sec:sumrule}--\ref{sec:material} build and stress-test the protocol. The driven
route is quantified, and ruled out by twelve orders, in Appendix~\ref{app:driven}.

\section{The Brillouin zone as a manifold without boundary}
\label{sec:hodge}

One structural fact underlies everything below: the Brillouin zone has no boundary. In
two dimensions it is a torus $T^2$, in three a torus $T^3$. Any total derivative
integrated over it vanishes, and that single statement is used repeatedly.

The Berry connection $A_\mu(\mathbf k)=i\langle u_{\mathbf k}|\partial_\mu u_{\mathbf k}\rangle$
is a one-form on that torus, and the Berry curvature
$\Om = \partial_x A_y-\partial_y A_x$ a two-form. In two dimensions a two-form is of top
degree, and this deserves a pause, since it is responsible for a dimensional accident
that costs people time. A $p$-form on a $d$-dimensional manifold decomposes as
\begin{equation}
\label{eq:hodge}
\Om \;=\; \bar\Om \;+\; \dd\alpha \;+\; \delta\gamma,
\end{equation}
into harmonic, exact and co-exact pieces, mutually orthogonal in $L^2$. In two dimensions
the co-exact piece of a two-form is empty, because there is no three-form for $\gamma$ to
be. Only in three dimensions, where the curvature is a two-form on a three-torus, do all
three sectors appear, and the co-exact sector is then the one carrying monopole sources,
that is, Weyl points.

We do not prove the decomposition. What matters is what the pieces are on a torus. The
harmonic piece is the constant, and its integral is quantized:
\begin{equation}
\label{eq:harmonic}
\bar\Om \;=\; \frac{2\pi C}{A_\BZ},\qquad
C \;=\; \frac{1}{2\pi}\int_\BZ \Om\,\dd^2k \;\in\;\mathbb{Z},
\end{equation}
with $A_\BZ$ the zone area. By Stokes' theorem the remaining piece $\dd\alpha$ integrates
to zero over the zone. The split into ``constant plus zero-mean remainder'' is thus the
split into ``topological plus geometric,'' and it is unique.

For a multiband system the decomposition is done band by band, and this is where the
paper begins. Each band $\lambda$ carries its own harmonic constant
\begin{equation}
\label{eq:bandharmonic}
\bar\Om_\lambda \;=\; \frac{2\pi C_\lambda}{A_\BZ},
\end{equation}
and the constants generally differ. At a Fermi surface only one band is in play for
electrons, and the difference is invisible. For bosons every band is populated, and it is
not.

\section{Which sectors the probes see}
\label{sec:visibility}

Almost every response of interest is a curvature integrated against a weight fixed by the
probe,
\begin{equation}
\label{eq:response}
\mathcal R \;=\; \int_\BZ w(\mathbf k)\,\Om(\mathbf k)\,\dd^2k .
\end{equation}
For the anomalous Hall conductivity $w$ is the occupation; for the nonlinear Hall response
it involves $\partial_\varepsilon f_0$; for the thermal Hall conductivity it is the Bose
weight $c_2$ defined below. On a manifold without boundary, inserting the decomposition
and integrating by parts gives
\begin{equation}
\label{eq:lemma}
\mathcal R \;=\; \bar\Om\int_\BZ w \;-\; \int_\BZ \dd w\wedge\alpha .
\end{equation}

\begin{lemma}[sector visibility]
\label{lem:visibility}
The harmonic sector couples only to the integrated weight $\int w$; the exact sector
couples only to the variation $\dd w$ of the weight.
\end{lemma}

Three corollaries follow at once, and it matters to be clear about which are new.

(i) If $\int w=0$ the harmonic sector drops out. This is the electronic noise theorem
discussed below, in which particle-number conservation pins the weight to zero mean.

(ii) If $\dd w=0$ the exact sector drops out, leaving only the Chern numbers. That is the
leading high-temperature behavior of the magnon thermal Hall effect, where
$c_2\to\pi^2/3$ uniformly; but it vanishes whenever $\sum_\lambda C_\lambda=0$, and the
limit is then set by the first correction, in which $\dd w\neq0$. Mook, Henk and Mertig
reach that correction by l'H\^opital's rule and obtain the first moment
$\sum_\lambda\int\varepsilon_\lambda\Om_\lambda$ as a figure of merit for the effect,
without the decomposition~\cite{MookHenkMertig2014}.

(iii) Where the curvature vanishes on the support of the weight, the exact sector must cancel
the harmonic one, since $\Om=\bar\Om+\dd\alpha$ with $\bar\Om$ constant. This explains the
low-temperature blindness of the thermal Hall effect, but it restates $\Om(\Gamma)=0$
rather than adding to it.

The maneuver common to all three, integration by parts to separate a quantized from a
non-quantized piece, is Haldane's~\cite{Haldane2004}. Lemma~\ref{lem:visibility} unifies
rather than discovers, and we use it as a design rule: \emph{the geometric sector is
visible exactly to the extent that the probe's weight varies over the zone}. That rule is
what renders Sec.~\ref{sec:parti} predictable in retrospect.

\section{The electronic theorem and its pair of hypotheses}
\label{sec:electronic}

For electrons, at second order in an applied field and under a conservation law, the
transverse current fluctuation has the structure
\begin{equation}
\label{eq:elecvar}
\Var(J_\perp)\;\propto\;\Var_p(\Om),
\end{equation}
a variance rather than a second moment, with $p$ a normalized thermal weight. The
fermionic counterpart of this object, an occupation-variance-weighted square of the
anomalous velocity, has been derived for time-reversal-invariant systems and named the
quantum fluctuation of the Berry curvature~\cite{WeiEtAl}; its weight carries no
conservation constraint, which is precisely hypothesis~(H1) below. The variance is what
carries the content: a variance annihilates a constant, the harmonic sector \emph{is} a
constant, and so the Chern number cancels out of the noise: the topological response is
silent in the fluctuations while the geometric sector is the noise source.

The proof requires two ingredients, of which only one is usually stated.

\begin{enumerate}
\item[(H1)] A conservation law, ensuring that the fluctuation weight has zero mean over
the zone, so that Lemma~\ref{lem:visibility} removes the harmonic coupling.
\item[(H2)] That the harmonic sector is a \emph{single} constant over everything the
weight touches.
\end{enumerate}

(H1) is the hypothesis everyone writes down. In the electronic setting (H2) is
invisible, because the Fermi surface selects essentially one band, and a single band has
a single $\bar\Om$. It is nonetheless a hypothesis, and the remainder of this paper is
what happens when it fails.

\section{Bosons: what is different}
\label{sec:bosons}

Phonons and magnons carry Berry curvature on their Brillouin zones, and their thermal
Hall conductance is a temperature-weighted integral of it over \emph{all} bands,
\begin{equation}
\label{eq:kxy}
\kxy \;=\; -\frac{\kB^{2}T}{\hbar V}\sum_\lambda\int_\BZ
c_2\!\big(n_B(\varepsilon_{\lambda\mathbf k})\big)\,\Om_\lambda(\mathbf k)\,\dd^2k ,
\end{equation}
with $n_B$ the Bose function and
\begin{equation}
\label{eq:c2}
c_2(n)\;=\;(1+n)\Big[\ln\frac{1+n}{n}\Big]^{2}-\ln^{2}n-2\,\mathrm{Li}_2(-n).
\end{equation}
Two properties of $c_2$ will matter. At high temperature it saturates, $c_2\to\pi^2/3$,
which is why the $\pi^{2}/3$ term in the thermal Hall response dies through the Chern
sum rule $\sum_\lambda C_\lambda=0$, itself a theorem here rather than an assumption
(Sec.~\ref{sec:partii}), leaving the $1/T$ correction to fix a finite high-temperature
limit (Appendix~\ref{app:sign}); and it obeys
a reflection identity under $n\to-(1+n)$, equivalent to $\varepsilon\to-\varepsilon$,
namely
\begin{equation}
\label{eq:c2reflect}
c_2(n)+c_2\big(-(1+n)\big)\;=\;\tfrac{2\pi^{2}}{3},
\end{equation}
which we have verified to thirty digits. Equation~\eqref{eq:c2reflect} states that
$c_2-\pi^2/3$ is odd in $\varepsilon/T$, so \emph{every} even moment is missing from the
high-temperature asymptotic series, to all orders. In particular the expansion runs
\begin{equation}
\label{eq:c2exp}
c_2 \;=\; \frac{\pi^{2}}{3}-\frac{\varepsilon}{T}
+\frac{1}{36}\Big(\frac{\varepsilon}{T}\Big)^{3}+\cdots ,
\end{equation}
with no quadratic term, so the coefficient after the leading one is the third moment of
the curvature-weighted density of states, not the second.

Three things change relative to the electronic case: there is no Fermi surface, so the
exact sector cannot mean ``Fermi-surface geometry''; bosons do not in general conserve
number, so (H1) must be re-established rather than assumed; and every band is populated,
so (H2) stands in jeopardy.

We work with collinear ferromagnetic magnons, for which the first difficulty does not
arise: $S_z$ is conserved at harmonic order, magnon number is conserved, the Hamiltonian
is Hermitian in the ordinary sense, and no paraunitary machinery is needed. Bogoliubov
systems (antiferromagnets and anything with anomalous pairing) do not conserve number
and require the symplectic quantum geometric tensor~\cite{TesfayeEckardt}; they fail (H1)
outright and we do not treat them here. Zyuzin and Kovalev's condition for a magnon
current to be well defined~\cite{ZyuzinKovalev},
$\hat O\sigma_3\hat H-\hat H\sigma_3\hat O=0$, is exactly (H1), identified independently
in the magnon setting, and the collinear ferromagnet is its trivial case.

\section{The fluctuating Boltzmann equation}
\label{sec:langevin}

We now derive the noise. The derivation applies standard kinetic fluctuation theory, and
we flag its provenance as we go.

Let $n_a$ denote the occupation of mode $a=(\lambda,\mathbf k)$ and $\delta n_a$ its
fluctuation about the equilibrium value $\bar n_a$. Write
\begin{equation}
\label{eq:Dweight}
\mathcal D_a \;=\; \bar n_a(1+\bar n_a)
\end{equation}
for the bosonic equal-time fluctuation weight. Take a number-conserving relaxation operator
\begin{equation}
\label{eq:relax}
\mathbb{L}\;=\;\tau^{-1}\big(\mathbbold{1}-\Pi\big),\qquad
\Pi_{ab}\;=\;\frac{\mathcal D_a}{\sum_c \mathcal D_c},
\end{equation}
which relaxes $\delta n$ toward a chemical-potential shift rather than toward zero. Then,
with $u$ the uniform vector, $u^{\mathsf T}\mathbb{L}=0$, which is the statement of number
conservation. The canonical correlator
\begin{equation}
\label{eq:correlator}
\mathbb{C}\;=\;\operatorname{diag}(\mathcal D)-\frac{\mathcal D \mathcal D^{\mathsf T}}{\sum_c \mathcal D_c}
\end{equation}
satisfies $\mathbb{L}\mathbb{C}=\mathbb{C}/\tau$; the Lyapunov condition
$\mathbb{L}\mathbb{C}+\mathbb{C}\mathbb{L}^{\mathsf T}=2\mathbb{J}$ holds exactly with
Langevin strength $\mathbb{J}=\mathbb{C}/\tau$; and $\mathbb{C}u=0$. All four have been
checked to machine precision. Integrating the decay of the transverse current correlator
gives
\begin{equation}
\label{eq:S0}
S(0)\;=\;2\tau\Big(\sum_c \mathcal D_c\Big)\,\Var_p(\Om),\qquad p=\mathcal D\Big/\!\sum_c \mathcal D_c .
\end{equation}

Equation~\eqref{eq:correlator} is not ours. The one-time correlation function of
occupation numbers given by Gantsevich, Gurevich and Katilius~\cite{GGK} has exactly this
structure: a rank-one subtraction $-N(\partial_N\bar F_p)(\partial_N\bar F_{p_1})$
carrying the property $\sum_{p_1}\varphi_{pp_1}=-\bar F_p$, which they state corresponds
to the particle-number conservation condition
$\sum_{p_1}(\delta F_p\,\delta F_{p_1})_{q\to0}=0$. That is the same rank-one
subtraction, the same annihilation of the uniform mode, and the same origin in number
conservation as~\eqref{eq:correlator} with $\mathbb{C}u=0$. The Boltzmann--Langevin
equation itself they attribute to Kogan and Shul'man~\cite{KoganShulman}. This section
therefore applies their framework rather than deriving it, and the verification reported
above confirms an implementation rather than establishing a result.

What matters for us is that the variance structure of~\eqref{eq:S0} is the genuine steady
state and not an ansatz, and that it \emph{emerges} from the conservation law rather than
being imposed.

\section{Where (H2) fails, and what replaces it}
\label{sec:failure}

Insert the band-resolved decomposition $\Om_a=\bar\Om_\lambda+\delta\Om_a$ into
Eq.~\eqref{eq:S0}. Since the harmonic constants are band-dependent,
Eq.~\eqref{eq:bandharmonic}, the variance splits into three pieces:
\begin{equation}
\label{eq:varsplit}
\Var_p(\Om)\;=\;\underbrace{\Var_p(\bar\Om_\lambda)}_{\text{topological}}
\;+\;\Var_p(\delta\Om)\;+\;2\Cov_p(\bar\Om_\lambda,\delta\Om),
\end{equation}
and the first term is nonzero whenever the Chern numbers differ. At a Fermi surface it
vanishes for electrons because there is only one $\bar\Om$; for bosons it does not.

\begin{proposition}
\label{prop:main}
In a multiband bosonic system the topological sector is not silent in the transport
noise. It enters through the occupation-weighted dispersion of Chern numbers across
populated bands, and it vanishes only if all populated bands share a Chern number.
\end{proposition}

The quantity is a variance of quantized numbers weighted by thermal populations. We
searched the literature for it as an observable and did not find it: what is tracked
there are per-band Chern numbers, their reordering across topological transitions,
bosonic Bott indices, and the occupation-weighted \emph{mean} through the thermal Hall
conductivity, but not the variance.

Table~\ref{tab:shares} gives the magnitude for the kagome ferromagnet with
Dzyaloshinskii--Moriya coupling, Chern numbers $(+1,0,-1)$, in the anchored convention of
Appendix~\ref{app:sign}. The topological share grows with temperature because the
dispersion can be seen at all only when the upper bands, which carry the nonzero Chern
numbers, are populated.

\begin{table}[t]
\caption{Topological share of the noise variance for the kagome ferromagnet, computed
from the pointwise Kubo curvature with a Zeeman gap $h=0.3\,JS$. Curvatures averaged
over plaquettes inflate these numbers by roughly a factor of three, because the
averaging smooths the near-degeneracy peaks that dominate $\langle\Om^2\rangle$; the
values below are converged, being identical to four figures at grid sizes $N=48$, $72$
and $108$.}
\label{tab:shares}
\begin{ruledtabular}
\begin{tabular}{lccc}
$D/J$ & $T=0.5\,JS$ & $T=1\,JS$ & $T=3\,JS$ \\
\colrule
$0.20$ & $0.53\%$ & $2.95\%$ & $9.01\%$ \\
$0.15$ & $0.38\%$ & $1.84\%$ & $5.44\%$ \\
\end{tabular}
\end{ruledtabular}
\end{table}

The share scales roughly as $D^2$: at small Dzyaloshinskii--Moriya coupling the curvature
develops peaks of height $\sim D^{-2}$ over a region of area $\sim D^{2}$, which leaves
the Chern numbers fixed but makes $\langle\Om^2\rangle$ grow as $D^{-2}$.

The gapless case is pathological rather than merely divergent. When the fluctuation
weight $\bar n(1+\bar n)$ diverges at $\Gamma$, where the curvature vanishes, the
normalized weight collapses onto that point and both the variance and the second moment
go to zero. A Zeeman gap is thus a requirement rather than a convenience, and it is in
any case the experimental situation: Hirschberger \emph{et al.}~\cite{Hirschberger} infer
$g\approx1.6$ from the exponential suppression by the gap $\Delta=g\mu_BB$, and for
Cu(1,3-bdc) a gap of $0.3\,JS=0.09$~meV needs about $1$~T, well inside the range they
sweep.

\section{The collision operator: a Landau--Placzek decomposition of the Berry curvature}
\label{sec:collision}

Equation~\eqref{eq:S0} assumed a single relaxation time. A physical magnon gas has at
least three rates: fast in-band redistribution conserving number and energy alike at
$1/\tau_C$, energy exchange with the phonon bath at $1/\tau_E$, and number relaxation at
$1/\tau_N$. Solving the corresponding Boltzmann--Langevin equation exactly yields three
Lorentzians whose weights are the spectral decomposition of $\Om$ against the conserved
quantities in the $\mathcal D$-weighted inner product:
\begin{equation}
\label{eq:threemode}
\big\langle\Om^{2}\big\rangle_p
=\underbrace{\langle\Om\rangle_p^{2}}_{\text{number mode}}
+\underbrace{\frac{\Cov_p(\Om,\varepsilon)^{2}}{\Var_p(\varepsilon)}}_{\text{energy mode}}
+\underbrace{\text{residual}}_{\text{fast bath}} .
\end{equation}
This closure has been verified to machine precision, with shares
$(2.33,15.92,81.75)\%$ at $T=0.5\,JS$, $(2.10,5.80,92.09)\%$ at $T=1$ and
$(1.62,0.54,97.84)\%$ at $T=3$, in two independent implementations.
Figure~\ref{fig:shares} shows the three shares across the temperature range for both
geometric vertices, the curvature treated here and the magnetization vertex taken up in
Sec.~\ref{sec:magnetization}.

Equation~\eqref{eq:threemode} is Gram--Schmidt in disguise: project $\Om$ onto the span
of the conserved quantities $\{1,\varepsilon\}$ in the $\mathcal D$-weighted inner product and
read off the squared components. The general statement, that for any collision operator
with conserved set $\{Q_i\}$ the slow spectral weights are the squared projections of the
observable onto $\operatorname{span}\{Q_i\}$, is the standard structure of fluctuating
hydrodynamics, and we present it as such. Mori's projector for a set of relevant
variables is $PX=(A_j,A_i)^{-1}(X,A_i)A_j$~\cite{Mori1965,teVrugtWittkowski}, which is
exactly this projection; Chaikin and Lubensky state as a general principle that there is
one mode, one peak in a response function, for each conserved
variable~\cite{ChaikinLubensky}; Kadanoff and Martin divide the total weight of the
density response into a proportion $1-c_v/c_p$ from diffusion and $c_v/c_p$ from
sound~\cite{KadanoffMartin}, a division Mountain presents as the Landau--Placzek
ratio~\cite{Mountain}; and R\'esibois supplies the kinetic version, the linearized
Boltzmann equation as an eigenvalue problem whose zero-eigenvalue eigenfunctions are the
collision invariants~\cite{Resibois}.

We therefore propose the name: Eq.~\eqref{eq:threemode} is the \emph{Landau--Placzek
decomposition of the Berry curvature}. The novelty lies not in the decomposition but in
the observable being decomposed, and in the identification of the number-mode weight
$\langle\Om\rangle_p^2$ with the occupation-weighted Chern moment.

Two consequences follow. An intermediate spectral feature appears at $1/\tau_E$ whose
weight is the energy-locked component of the curvature; note that
$\Cov_p(\Om,\varepsilon)$ is an energy-curvature correlation of the same kind that
builds the high-temperature coefficient of Eq.~\eqref{eq:c2exp}, but not the same object,
since that coefficient is a bare zone integral while this is $\mathcal D$-weighted. And the rate
ratio enhances the topological step, the zero-frequency weight of the topological
Lorentzian being $\tau_N\langle\Om\rangle_p^2$ against $\tau_C\Var$ for the plateau.
Table~\ref{tab:enhance} gives the resulting topological share of $S(0)$.

\begin{table}[t]
\caption{Topological share of $S(0)=\sum_i\tau_iw_i$ as a function of the rate ratio, at
$\tau_E=\tau_C$. A realistic ratio of order thirty already raises the share from the
naive one or two per cent into the tens; dominance requires ratios in the hundreds.}
\label{tab:enhance}
\begin{ruledtabular}
\begin{tabular}{lccccccc}
$\tau_N/\tau_C$ & $1$ & $10$ & $30$ & $60$ & $100$ & $150$ & $238$\\
\colrule
$T=0.5\,JS$ & $2.3\%$ & $19.2\%$ & $41.7\%$ & $58.8\%$ & $70.4\%$ & $78.1\%$ & $85.0\%$\\
$T=1\,JS$   & $2.1\%$ & $17.7\%$ & $39.2\%$ & $56.3\%$ & $68.2\%$ & $76.3\%$ & $83.6\%$\\
$T=3\,JS$   & $1.6\%$ & $14.1\%$ & $33.0\%$ & $49.7\%$ & $62.2\%$ & $71.2\%$ & $79.7\%$\\
\end{tabular}
\end{ruledtabular}
\end{table}

\begin{figure*}[t]
\includegraphics[width=\textwidth]{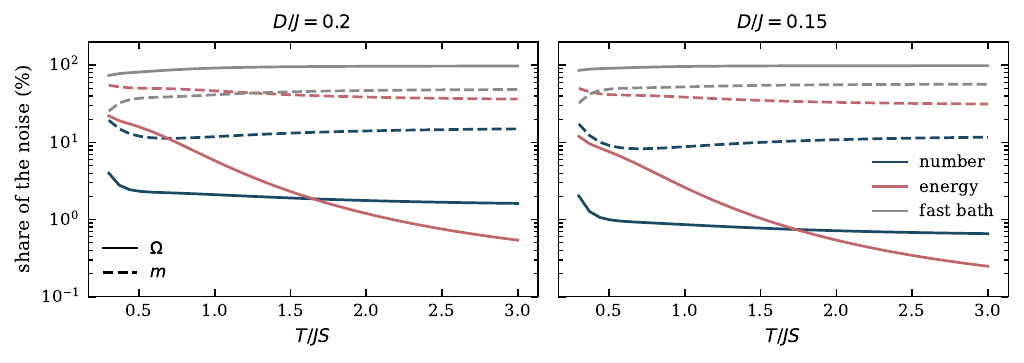}
\caption{Landau--Placzek decomposition of the equilibrium noise of the two geometric
vertices, at the two couplings of interest. Solid curves are the Berry curvature $\Om$,
dashed the self-rotation moment $m$; the three modes are the number-conserving,
energy-conserving and fast-bath channels of Sec.~\ref{sec:collision}. At every
temperature the magnetization noise is far less bath-dominated than the curvature noise,
because $m$ carries one energy denominator against the curvature's two, and the
separation persists across the coupling range.}
\label{fig:shares}
\end{figure*}

\section{The energy-magnetization subtraction at the level of fluctuations}
\label{sec:magnetization}

Before any of this can be called an observable, a trap in thermal transport must be
dealt with. Qin, Niu and Shi~\cite{QinNiuShi} showed that applying the Kubo formula
directly to thermal transport yields an unphysical divergence, because a temperature
gradient drives not only a transport heat current but also a circulating heat current
that no transport experiment sees. Obtaining the observable coefficient requires
subtracting the gravitomagnetic energy magnetization, and their treatment covers bosons
as well as fermions.

At the level of \emph{means} the subtraction is benign in a specific sense: the weight
changes, the form does not. That is exactly the content of Eq.~\eqref{eq:kxy}, whose
weight $c_2$ is what emerges once the subtraction has been performed. The subtraction
therefore acts diagonally on the Hodge sectors at the mean level, and since $c_2\ge0$ has
nonzero per-band zone integral at every temperature, it does \emph{not} remove the
harmonic sector from the observable mean. Harmonic silence in the noise is consequently
a theorem to prove, not a definition.

At the level of \emph{fluctuations} the question has not previously been posed, and its
answer is structural rather than computational: for the noise a transport measurement
records, nothing needs subtracting.

At leading semiclassical order the local number-current density of the magnon gas
consists of a transport vertex plus a curl~\cite{XiaoChangNiu},
\begin{widetext}
\begin{equation}
\label{eq:localcurrent}
\mathbf j(\mathbf r,t)=\sum_a\Big[\mathbf v_a+\tfrac{1}{\hbar}\,\Om_a\,
\hat{\mathbf z}\times\mathbf F\Big]\,n_a(\mathbf r,t)
\;+\;\nabla\times\big[\mathcal M(\mathbf r,t)\,\hat{\mathbf z}\big],
\qquad
\mathcal M=\sum_a m_a\,n_a ,
\end{equation}
\end{widetext}
with $\mathbf v_a$ the group velocity, $\mathbf F$ a mechanical force, the momentum
integral implied in the mode sum, and
\begin{equation}
\label{eq:selfrot}
m_\lambda(\mathbf k)=-\frac{1}{2\hbar}\,
\mathrm{Im}\,\big\langle\partial_{\mathbf k}u_\lambda\big|\times
\big(H(\mathbf k)-\varepsilon_\lambda\big)\big|\partial_{\mathbf k}u_\lambda\big\rangle
\cdot\hat{\mathbf z}
\end{equation}
the self-rotation moment of the magnon
wavepacket~\cite{MatsumotoMurakami1,MatsumotoMurakami2}, the number-current analogue of
the orbital moment. What a transport measurement records is the current collected
between reservoirs, equivalently the current through a complete interior cross-section,
$I(t)=\int\dd y\,j_x(x_0,y,t)$, with all densities vanishing outside the sample.

\begin{theorem}[vanishing of the fluctuation-level subtraction]
\label{thm:vanish}
For any bounded magnetization density that vanishes outside the sample,
\begin{equation}
\int\dd y\,\big[\nabla\times(\mathcal M\hat{\mathbf z})\big]_x(x_0,y,t)
=\int\dd y\,\partial_y\mathcal M=0
\end{equation}
for every complete cross-section and every instant. Magnetization fluctuations therefore
contribute exactly zero to lead-measured current noise, realization by realization, at
every frequency and to all orders in the drive.
\end{theorem}

\begin{proof}
The fundamental theorem of calculus together with the boundary condition. In Fourier
form the section current involves only $q_y=0$, where the $x$ component of
$i\mathbf q\times\hat{\mathbf z}\,\widetilde{\mathcal M}$, namely
$iq_y\widetilde{\mathcal M}$, vanishes identically; no small-$q$ expansion enters.
\end{proof}

Every driven result above is therefore left unmodified: the constrained
correlator~\eqref{eq:correlator}, the variance~\eqref{eq:S0}, the
decomposition~\eqref{eq:threemode} with its rate-ratio enhancement, and
Proposition~\ref{prop:main} itself. The fluctuation weight remains $\mathcal D=n(1+n)$, whose
per-band zone integral is nonzero at every temperature, so the harmonic non-silence
promised by the mean-level analysis is physical and final. At the mean level the matching
statement, that circulating currents deliver no net flow through a complete
cross-section, is due to Cooper, Halperin and Ruzin~\cite{CooperHalperinRuzin}, and the
underlying ambiguity, that a curl may always be moved between the transport and
magnetization currents without touching the conservation law, is stated in general form
by Kapustin and Spodyneiko~\cite{KapustinSpodyneiko};
Theorem~\ref{thm:vanish} is its realization-by-realization extension, which we have not
found stated for noise.

The theorem coexists with the mean-level subtraction because the two statements concern
different objects. The Qin--Niu--Shi contamination lives in gauge-represented response
functions: a statistical drive cannot be written as a force, so it is represented by an
auxiliary field, Luttinger's gravitational potential~\cite{Luttinger} or Tatara's
thermal vector potential~\cite{Tatara}, whose Kubo response in the transport order of
limits, $\omega\to0$ before $q\to0$, contains magnetization terms absent from the
measured coefficient. Any noise formula obtained by fluctuation-dissipation manipulation
of that gauge-represented response inherits them. The kinetic route of
Sec.~\ref{sec:langevin} computes the correlator of the physical section current directly
and never enters that gauge. For bosons the same separation is known at the mean level,
where the kinetic route isolates the bound magnetization currents transparently and
needs no pseudogravitational field~\cite{MangeolleSavaryBalents,KuwabaraNasu}, and the
same distinction is why the sum rule of Sec.~\ref{sec:sumrule} is untouched, as discussed
there.

For the equilibrium observable a stronger statement holds, one that removes the
detection geometry from the list of assumptions.

\begin{lemma}[chiral protection of the antisymmetric cross-spectrum]
\label{lem:chiral}
In Fourier space the curl current of a single scalar density is rank one,
$\mathbf j^{\mathcal M}(\mathbf q)=\mathbf w(\mathbf q)\,\widetilde{\mathcal M}(\mathbf q)$
with $\mathbf w=i(q_y,-q_x)$, so its cross-spectral matrix is
$w_iw_j^{*}\,S_{\mathcal M\mathcal M}$ and
$S_{xy}-S_{yx}=2i\,\mathrm{Im}(w_xw_y^{*})\,S_{\mathcal M\mathcal M}$ with
$w_xw_y^{*}=-q_xq_y$ real. At every wavevector and frequency, under any drive and any
statistics, the magnetization auto-correlation contributes exactly zero to $S^{A}$: a
single scalar cannot supply a chiral correlation.
\end{lemma}

The statement survives arbitrary detection: for
$J_i=\int K_i(\mathbf r)\,j_i(\mathbf r)\,\dd^2r$ with distinct pickup kernels and a
stationary homogeneous magnetization correlator that is even in $\mathbf q$ at each
frequency, the magnetization part of $S_{xy}-S_{yx}$ cancels, the kernel factor being
antisymmetric under exchange while the mixed second derivative of an even correlator is
even. That evenness is a hypothesis, not a consequence of stationarity: a drifting
medium, whose correlator depends on $\mathbf u-\mathbf v\tau$, is even at equal times and
at no other lag, and does contribute. On the inversion-symmetric lattice in equilibrium
evenness holds, which is the case the paper deploys. Whatever the microwave detection
realizes (a complete section, an interfacial pickup or a local probe at standoff), no
magnetization noise enters $S^{A}$. The sole surviving magnetization term there is the
transport-magnetization cross term, which vanishes at $q=0$ by parity on the
inversion-symmetric lattice, $\mathbf v$ odd and $m$ even under $\mathbf k\to-\mathbf k$,
verified on the model to $10^{-19}$, and is expected at $O(q^{2})$ through the streaming
term, a counting we flag for verification rather than assert. The lead measurement of
the driven protocol (Appendix~\ref{app:driven}) is a complete section by construction, so
Theorem~\ref{thm:vanish} covers it exactly.

What the magnetization does carry is finite-$q$ structure, fixed by two facts. In local
equilibrium at a mechanical potential $U(\mathbf r)$, with
$G(x)=\kB T\ln(1-e^{-x/\kB T})$ so that $G'=n_B$, the anomalous term of
Eq.~\eqref{eq:localcurrent} is itself a curl,
\begin{widetext}
\begin{equation}
\label{eq:LEidentity}
\sum_a\tfrac{1}{\hbar}\,\Om_a\,(\hat{\mathbf z}\times\mathbf F)\,
n_B(\varepsilon_a+U)
=\nabla\times\big[\mathcal M^{\rm stat}\hat{\mathbf z}\big],\qquad
\mathcal M^{\rm stat}=\tfrac{1}{\hbar}\sum_a\Om_a\,G(\varepsilon_a+U),
\end{equation}
\end{widetext}
which is the statistical Berry magnetization of the Bose gas~\cite{XiaoShiNiu}. The
statistical part is thus not an independent fluctuating object under mechanical
potentials, and the only independent magnetization vertex is the self-rotation moment of
Eq.~\eqref{eq:selfrot}, whose noise follows from the projection structure of
Sec.~\ref{sec:collision} with $m$ in place of $\Om$. Its character differs instructively:
at $T=0.5\,JS$ and $D/J=0.2$ the (number, energy, residual) shares of
$\langle m^{2}\rangle_p$ are $(12.1,50.6,37.4)\%$ against $(2.3,15.9,81.8)\%$ for the
curvature, spectrally converged, because $m$ carries one energy denominator against the
curvature's two and so correlates far more strongly with $\varepsilon$; a local-probe
admixture rides the slow Lorentzians, and by Lemma~\ref{lem:chiral} none of it enters
$S^{A}$. The contrast is not an artifact of the coupling strength: at the material-like
$D/J=0.15$ the curvature shares fall to $(1.0,7.6,91.4)\%$ while the magnetization shares
shift only to $(9.1,41.7,49.2)\%$, so across the range of interest the two vertices sit
on opposite sides of the bath-dominated regime. The two-band identity
$m_{\rm lower}=\tfrac{1}{2\hbar}(\varepsilon_2-\varepsilon_1)\Om_{\rm lower}$ supplies
the arbiter row for the new vertex (Appendix~\ref{app:sign}).

Scope, stated plainly. Equation~\eqref{eq:localcurrent} is the leading semiclassical
order, which is the order of every result in this paper, and its anomalous vertex treats
the drive-induced interband coherence as adiabatically slaved, valid well below the
minimal direct gap; the driven Lorentzians live at collision rates and satisfy this, and
the sum rule makes no use of the vertex. A statistical fluctuation of $\mu$ or $T$ shifts
the local distribution through the same $G$-rearrangement formally but exerts no force,
so the anomalous vertex does not act on it and, at this order, the current functional
contains no statistical-magnetization fluctuation term; its finite-$q$ form belongs,
together with the fluctuation theory of statistical drives itself, to the heat-channel
continuation of this work on the nonequilibrium framework of Ref.~\cite{GGK}.

\section{The equilibrium sum rule}
\label{sec:sumrule}

A driven measurement pays a suppression penalty, because the anomalous velocity requires
a force; Appendix~\ref{app:driven} quantifies the penalty and rules the driven route out.
It can be avoided altogether.

For the number-conserving ferromagnet with no drive, the antisymmetric part of the
symmetrized current cross-spectrum, $S^{A}=\frac{1}{2i}(S_{xy}-S_{yx})$, obeys
\begin{eqnarray}
-\frac{2}{\pi}\int_{0}^{\infty}\frac{\dd\omega}{\omega^{2}}\,
\tanh\!\Big(\frac{\beta\omega}{2}\Big)\,S^{A}(\omega)
&=&\sum_\lambda\int_\BZ\!\frac{\dd^{2}k}{(2\pi)^{2}}\,n_\lambda\Om_\lambda
\nonumber \\
&\equiv&\langle\Om\rangle_n ,
\label{eq:sumrule}
\end{eqnarray}
whose harmonic content is $(2\pi/A_\BZ)\sum_\lambda C_\lambda N_\lambda$, the
occupation-weighted first moment of the Chern numbers. The derivation employs the
fluctuation-dissipation pairing $S^{A}=\coth(\beta\omega/2)\chi''_{A}$, which survives
antisymmetrization because the Dzyaloshinskii--Moriya term breaks time reversal, together
with the band-pair curvature identity for the velocity matrix elements.
Equation~\eqref{eq:sumrule} has been checked numerically: it closes to $0.1\%$ at $T=0.5$
and $3\,JS$ and to $0.15\%$ at $T=1$.\footnote{The check is analytic before it is
numerical: pair by pair, the factor $\tanh(\beta\omega/2)$ converts the symmetrized
weight $\coth(\beta\omega/2)$ into the occupation difference $n_{\lambda}-n_{\lambda'}$,
so Eq.~\eqref{eq:sumrule} holds as an identity mode by mode before the $\omega$- and
$k$-integrations are performed. The $0.1\%$ closure therefore measures the accuracy of
the quadrature, not of the relation.}

What this relation is and is not deserves statement. Applying the
fluctuation-dissipation step, $\tanh(\beta\omega/2)S^{A}=\chi''_{A}$, the left-hand side
becomes a Kramers--Kronig transform of the absorptive antisymmetric response, and
combined with the standard curvature formula for that response the sum rule states that
equilibrium noise determines the same coefficient a driven magnon Hall measurement would.
That is what the fluctuation-dissipation theorem says in general.
Equation~\eqref{eq:sumrule} is therefore the bosonic, finite-temperature, current-noise
instance of the fermionic optical Hall sum rule, and it stands to it as
Lemma~\ref{lem:visibility} stands to Haldane. We present it as the consistency check that
makes the equilibrium route quantitative, not as a new relation. Its pedigree is worth
stating explicitly: it is the finite-temperature, occupation-weighted, multiband analogue
of the zero-temperature noise sum rule of Ref.~\cite{NeupertChamonMudry}, to which it
reduces when the thermal factor is replaced by its zero-temperature limit and one band is
filled; the novelty here is the weighting, not the relation.
Section~\ref{sec:magnetization} supplies the reason the relation is untouched by the
magnetization subtraction: it is built from the current commutator at $q=0$ and finite
$\omega$, where curls have no component and no drive representation or transport ordering
of limits enters, which is the Kadanoff--Martin separation of the correlator from the
transport limit~\cite{KadanoffMartin}.

The magnitude problem dissolves. There is no drive, hence none of the $\eta^{2}$
suppression of Appendix~\ref{app:driven} anywhere. The chirality ratio
$r=|S^{A}|/\sqrt{S^{S}_{xx}S^{S}_{yy}}$, bounded by unity through Cauchy--Schwarz on the
interband matrix elements, reaches $0.810$, $0.472$ and $0.529$
at the spectral peaks for $T=0.5$, $1$ and $3\,JS$ (at the model coupling): an order-one fraction
of the interband noise at its own
frequency, against the $10^{-10}$ of the driven case. The spectrum turns on at the
minimal direct gap $2\sqrt3\,D$, which for Cu(1,3-bdc) at $D/J=0.15$ is $0.52\,JS$ or
$38$~GHz, with the main weight near $2.0\,JS$, or $147$~GHz.

Two things the framing must not claim. This and the driven result are \emph{not} two
moments of one distribution: the sum rule weights by $n$, the driven variance by
$\mathcal D=n(1+n)$, so there are two thermal weightings. And the sum rule does not isolate the
topology. The harmonic part is $12.2$, $4.8$ and $3.2$ times the total at $T=0.5$, $1$
and $3\,JS$ (at the model coupling $D/J=0.2$), with the geometric part canceling most of it, so the measured quantity is a
small difference of two large sector contributions. This is corollary (iii) of
Lemma~\ref{lem:visibility} once more, and it motivates everything in the next two
sections.

\begin{figure}[t]
\includegraphics[width=\columnwidth]{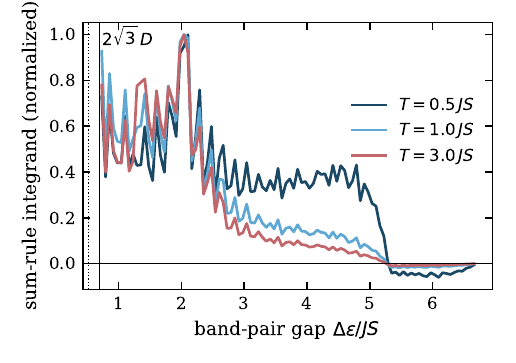}
\caption{The integrand of the sum rule, resolved in the band-pair gap coordinate and
normalized to its largest value, at three temperatures. What frequency resolution
resolves is exactly the coordinate the $\omega$ integration pools over, which is why the
integrated form cannot in practice separate the sectors while the resolved form can. The turn-on at
the minimal direct gap $2\sqrt3\,D$ is marked for $D/J=0.2$ (solid) and $0.15$ (dotted).}
\label{fig:integrand}
\end{figure}

\section{Why frequency resolution is needed: the integrated observable cannot separate
the sectors}
\label{sec:parti}

The natural first attempt at extraction exploits quantization. The harmonic content of
Eq.~\eqref{eq:sumrule} is $(2\pi/A_\BZ)\sum_\lambda C_\lambda N_\lambda$, with the
$C_\lambda$ integers, fixed and field-independent, while the $N_\lambda$ are calculable
populations depending on temperature and Zeeman field. A sweep of $(T,h)$ would then
trace a family of curves in which the harmonic part is a known linear form in a few
integers, and fitting to a lattice rather than a continuum ought to be well conditioned.

It fails, and the reason is instructive. Because the occupation depends on $\mathbf k$
only through $\varepsilon$, the entire sweep is a single integral transform of one
function,
\begin{align}
\label{eq:transform}
\langle\Om\rangle_n(T,h)&=\int\dd\varepsilon\,F(\varepsilon)\,
n_B\!\Big(\frac{\varepsilon+h}{T}\Big),\nonumber\\
F&=\sum_\lambda\Big[\tfrac{2\pi C_\lambda}{A_\BZ}\varrho_\lambda+g_\lambda\Big],
\end{align}
with $\varrho_\lambda$ the band density of states and $g_\lambda$ the zero-mean geometric
density. Both sectors sit inside $F$. The geometric contribution is therefore a linear
form over the \emph{identical} kernel family, and integrality buys nothing, because the
degeneracy directions are themselves integer vectors. Since $\int g_\lambda=0$ exactly,
there is a clean identity, $C_\lambda=2\pi\int_\lambda F$, so the triple is identifiable
in principle; but the protocol is then a deconvolution of $F$, and over the accessible
window the upper bands are Boltzmann, so $h$ enters only through a global factor
$e^{-h/T}$ and the two-dimensional sweep collapses to a one-dimensional Laplace
inversion.

The conditioning is quantitative. Against a complete zero-mean nuisance class, the
per-point identifiability signals relative to the observable RMS are $9.26\e{-5}$ for
$\delta C=(1,0,0)$, $7.60\e{-5}$ for $(0,1,0)$ and $(0,0,1)$, $2.75\e{-5}$ for
$(1,-1,0)$, and $6.53\e{-9}$ for $(0,1,-1)$, against a model error of $1.13\e{-5}$. The
direction moving one Chern unit between the two thermally starved upper bands sits some
seventeen hundred times \emph{below} the floor. In Monte Carlo, a rigid nuisance model
returns a stable, confident and wrong triple on every trial, and a flexible one selects
the sign-flipped triple as often as the truth, so even the overall chirality goes
undetermined.

The hypothesis deserves statement, because it is exactly what a companion literature does
not satisfy. The obstruction is the pooling of $\mathbf k$ into level sets of
$\varepsilon$, and it bites only when the zero-mean geometric densities $g_\lambda$ have
nontrivial support in that pooled coordinate, which requires dispersive bands. In the
perfectly flat limit each band's level set is the whole zone, $g_\lambda$ vanishes
identically, and the clean identity $C_\lambda=2\pi\int_\lambda F$ becomes directly
usable with no deconvolution: the frequency integral then does return an integer.
Kruchkov and Ryu~\cite{KruchkovRyu2024} realize that limit, their integrated current
noise returning $Ce^{2}\Delta^{2}$ for two flat Chern bands. Their route is
complementary to ours in two further respects worth recording. Their multiband
expression, their Eq.~(16) (equation numbers refer to the published Letter, whose
numbering sits one below the arXiv version), carries the pair gaps
$\Delta_{nm}^{2}(\mathbf k)$ \emph{inside} the momentum sum, so the gap leaves the
integral only in the constant-gap case of their Eq.~(22); for dispersive partners their
own formula is a gap-weighted geometric integral, which is the pooled object above. And
their measured channel is the symmetric one, converted to topology through the
ideal-flat-band trace condition $\mathrm{Tr}\,\mathcal G=|\mathcal F|$ of their Eq.~(20),
which returns the modulus and so cannot resolve a signed triple. The same authors state
the dispersive case in the conductivity language, where the frequency-integrated
longitudinal sum rule returns the quantum-metric and Luttinger invariants and no Chern
number~\cite{KruchkovRyu2023}, which is the present result viewed from the other side.

We record a methodological point rather than dropping it. This sweep is the
moment-tomography problem of Sec.~\ref{sec:sumrule} in a different parametrization, and
the effective rank of that kernel, three to four at realistic precision, was already
known to us. Extracting three Chern numbers plus the shape freedom in each band window
requires more functionals than that.

A second and independent reason the integrated observable cannot isolate the harmonic
content arrives from a direction unconnected with rank. Equation~\eqref{eq:sumrule}
transforms $\chi''_A$ over all frequencies, so what it returns is the full transverse
coefficient, the intrinsic part together with whatever disorder contributes; and for
bosonic thermal Hall transport the extrinsic contribution is not a correction but a term
of the same order, one that can survive where the intrinsic one
vanishes~\cite{MangeolleKnolle}. The frequency integral therefore pools the intrinsic
against the extrinsic channel exactly as it pools the harmonic against the geometric
sector, and for the same reason: it discards the coordinate that separates them. The two
channels separate in frequency rather than in $\mathbf k$, the extrinsic weight being
collisional and sitting near the origin while the intrinsic weight is interband and sits
at the band-pair gaps, the familiar division in the anomalous Hall effect, and our
reading rather than a result of Ref.~\cite{MangeolleKnolle}, whose calculation is of the
static coefficient. Frequency resolution defeats the second pooling for the same reason
it defeats the first, and Sec.~\ref{sec:discussion} returns to what that buys.

\section{The frequency-resolved protocol}
\label{sec:partii}

The $\varepsilon$-integration is precisely where the integrated route of
Sec.~\ref{sec:parti} dies, and frequency resolution undoes it. The mechanism is not new
in itself: for two bands at zero temperature the frequency-resolved noise already returns
the geometric tensor on the contour of constant direct gap, so resolving $\omega$
resolves the gap coordinate~\cite{NeupertChamonMudry}. Figure~\ref{fig:integrand} shows
what that resolution exposes: the integrand of the sum rule, spread over the band-pair
gap and turning on at $2\sqrt3\,D$, in place of the single number its $\omega$ integral
returns. What must be supplied here is everything that multiband finite-temperature
operation adds: a decomposition that survives more than one pair, an identifiability
analysis against a nuisance class rich enough to be honest, and a criterion for when the
integers are recoverable at all. Work with the sum-rule integrand before
$\omega$-integration, binned in the gap variable, built from the pair decomposition
$\Om_n=\sum_m\Om_{nm}$ with $\Om_{nm}=-\Om_{mn}$. The data then resolve the gap
coordinate directly rather than pooling it, and the transform acquires the rank the sweep
lacked.

Three structural facts come with the pair decomposition, all checked to machine precision
in two independent implementations. The pair curvatures reconstruct the band curvatures
to $8.9\e{-16}$. The pair contents are reals but, unlike the $C_\lambda$,
convention-dependent: $(1.0061,-0.0061,1.0061)$ in the atomic convention of
Appendix~\ref{app:sign} and $(0.9050,0.0950,0.9050)$ in the periodic one, the two
differing by $0.1011$ exactly along the gauge direction identified below, which is the
third trap of Appendix~\ref{app:sign} exhibited in pair space. Only the antisymmetric
combinations
\begin{equation}
\label{eq:pairmap}
C_1=c_{12}+c_{13},\quad C_2=-c_{12}+c_{23},\quad C_3=-c_{13}-c_{23}
\end{equation}
are quantized, giving $C=(1,0,-1)$; the direction $(1,-1,1)$ in pair space is annihilated
by that map, so it is a gauge, profiled out with the nuisance, which makes every quoted
number gauge-invariant. And $\sum_\lambda C_\lambda=0$ becomes a \emph{theorem} rather
than an imposed constraint, since the band-summed curvature vanishes pointwise by pair
antisymmetry, verified to $3.3\e{-16}$, so the candidate lattice is the traceless plane
by construction.

The reversal relative to the integrated route is qualitative. Growing the nuisance from
three to thirty-five zero-content modes per pair collapses the model error while the
topological directions merely plateau, so the ratio of signal to model error \emph{grows}
where in the integrated observable it fell. Table~\ref{tab:scissors} gives one
implementation and Fig.~\ref{fig:scissors} the same numbers as a curve; the other shows
the same growth with different values, and the difference between them is itself
informative (Sec.~\ref{sec:robust}).

\begin{table*}[t]
\caption{Identifiability against nuisance completeness, 48 bins, full grid, atomic
convention. The ratio is not monotone: it turns over near twenty-five modes and then
jumps by almost two orders once the nuisance basis is complete enough for the model error
to collapse. What the argument requires is the value at completeness, not a monotone
approach to it. The values are implementation-dependent; the two implementations agree on
the ratios.}
\label{tab:scissors}
\begin{ruledtabular}
\begin{tabular}{lcccc}
modes/pair & model error & chirality $(1,0,-1)$ & split $(0,1,-1)$ & split/error \\
\colrule
$3$  & $2.72\e{-1}$ & $7.27\e{-1}$ & $1.67\e{-1}$ & $0.6$\\
$9$  & $1.08\e{-1}$ & $4.04\e{-1}$ & $8.25\e{-2}$ & $0.8$\\
$17$ & $4.13\e{-2}$ & $2.70\e{-1}$ & $5.55\e{-2}$ & $1.3$\\
$25$ & $1.37\e{-2}$ & $3.89\e{-2}$ & $3.48\e{-3}$ & $0.3$\\
$35$ & $4.32\e{-5}$ & $1.79\e{-2}$ & $9.97\e{-4}$ & $23.1$\\
\end{tabular}
\end{ruledtabular}
\end{table*}

\begin{figure}[t]
\includegraphics[width=\columnwidth]{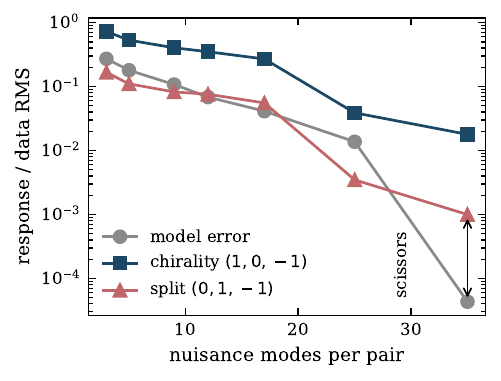}
\caption{Identifiability against nuisance completeness, the numbers of
Table~\ref{tab:scissors} plotted as a curve. Across the range the model error falls by
four orders while the two topological responses fall by little more than one, so the gap
between them opens; the ratio itself is not monotone, and what the argument uses is its
value once the basis is complete.}
\label{fig:scissors}
\end{figure}

\begin{figure*}[t]
\includegraphics[width=\textwidth]{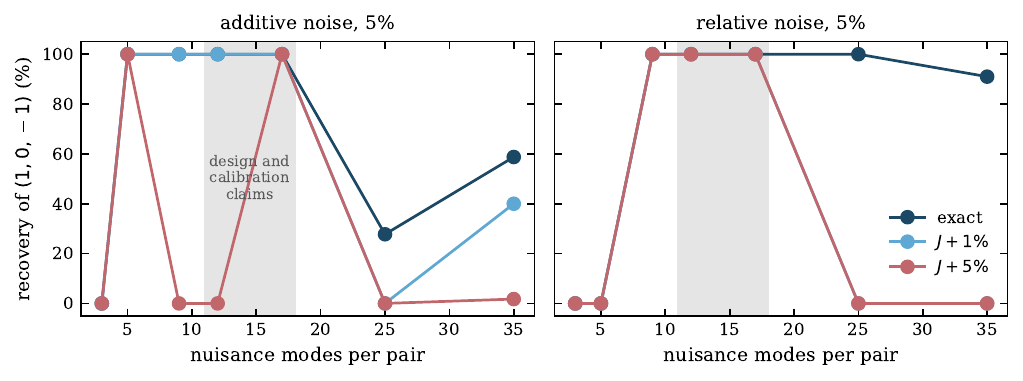}
\caption{Recovery of the true triple against nuisance mode count, for exact and
miscalibrated dispersion, under both noise models, at four hundred trials per point.
Recovery is not monotone in the mode count, and the exact and perturbed cases move in
opposite directions at adjacent counts. The shaded band is where the design and
calibration results of Sec.~\ref{sec:robust} are established; the identifiability
argument of Sec.~\ref{sec:partii} is established at thirty-five modes, outside it. Both
the structure and the separation follow from holding the resolution fixed at forty-eight
bins: Fig.~\ref{fig:gap} shows the same data as a function of $M_{2}/N_{\rm bins}$ and
removes the separation. We show the curve rather than a single count because at fixed
resolution the choice of count is a knob that moves the answer both ways, and a criterion
stated after the fact would not be a criterion.}
\label{fig:recovery}
\end{figure*}

\begin{figure*}[t]
\includegraphics[width=\textwidth]{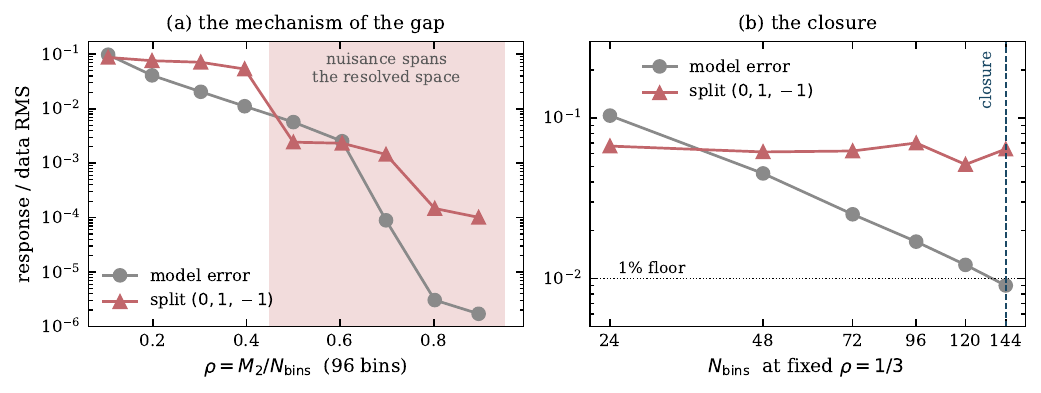}
\caption{The mode-count gap and its closure. (a) At fixed resolution the split response
collapses once the nuisance dimension approaches the resolved dimension in the gap
coordinate, near $\rho=M_{2}/N_{\rm bins}\simeq0.45$ (shaded): growing the nuisance then
buys model error by spending signal, and this is the mechanism behind the apparent
tension between the design results and the identifiability argument. (b) Along fixed
$\rho=1/3$ there is no trade. The model error falls roughly as $N_{\rm bins}^{-1.4}$ while
the split stays near $6\e{-2}$, so the separation opens by raising resolution. The one per
cent floor is reached at one hundred and forty-four bins, where the design reduction, the
identifiability argument and the calibration tolerance hold simultaneously. Full grid,
five per cent additive noise, forty trials per point on a $120\times120$ mesh (entries at
one hundred per cent carry a rule-of-three bound of $7.5\%$); the closure numbers quoted
in the text are at one thousand trials on the same mesh, where the model error at the
closure point is $9.0\e{-3}$, meeting the floor with about ten per cent of margin; the
eight closure cells at one thousand of one thousand carry the same rule-of-three bound,
here $0.3\%$, since a normal standard error is not meaningful at the boundary. Atomic
convention.}
\label{fig:gap}
\end{figure*}

\section{Robustness, design and error budget}
\label{sec:robust}

A protocol is only as good as its failure modes, which we now characterize.

\subsection{The one-bin control as an exact identity}

Collapsing to a single frequency bin reproduces the integrated observable of
Sec.~\ref{sec:parti} setting by setting, to $7\e{-15}$. One bin \emph{is} that observable,
verbatim, so the control is an identity rather than an analogy. At one bin the signal and
the model error alike fall to numerically meaningless levels, and their ratio is unstable
between implementations, reading $954$ in one and $320$ in the other. That
irreproducibility is the argument for an absolute criterion: feasibility is decided by
the signal against experimental precision, and the ratio to model error only checks that
the nuisance is not swallowing the signal. On the absolute criterion the one-bin signal
is about $10^{-10}$ of the data RMS, ten orders below any experimental precision.

\subsection{Estimator bias under incomplete nuisance}

With a rigid nuisance the systematic misfit exceeds the signal and projects onto the
candidate lattice, and \emph{which} wrong triple it favors depends on the nuisance basis.
At three modes per pair the zero-noise continuous fit of one implementation returns
$(1.167,-2.208,1.041)$ and recovers zero per cent; the other recovers a hundred. Neither
is a coding error. Consequently Monte Carlo recovery figures must be quoted only where
the model error lies below the stated precision floor, and the nuisance basis and mode
count must be reported beside every figure. We follow that rule throughout.

\subsection{The resolution requirement}

The split signal grows strongly with bin count and depends jointly on nuisance
completeness. The obvious prescription, to fix a precision floor $p$ and grow the nuisance
until the model error falls below $p$, is not safe on its own. The counterexample is the
nine-setting design at the largest mode count its data support, thirty-five modes per
pair at forty-eight bins, outside the matched-flexibility comparison of
Table~\ref{tab:design}: the model error is $1.2\e{-5}$, below any floor one would state,
while the split falls to $1.4\e{-4}$ and recovery collapses to a few per cent. With few
data the nuisance can reach the floor by absorbing the signal rather than by describing
the geometry. The remedy is to grow the nuisance and the resolution together at fixed
$\rho=M_{2}/N_{\rm bins}\lesssim1/3$, the condition that the nuisance dimension not
exceed the resolved dimension, and to raise $N_{\rm bins}$ until the model error falls
below $p$; along that line the split does not fall (Fig.~\ref{fig:gap}). The quantity to
check is the split against the \emph{aggregate} noise,
\begin{equation}
\label{eq:R}
R\;=\;\nu_{\rm split}\,\sqrt{P}\,/\,\sigma ,
\end{equation}
with $P$ the number of data points and $\sigma$ the per-point noise. Across twenty
configurations spanning four designs, recovery is essentially complete wherever $R$
exceeds about four, partial between one and three, and collapses below one. The
transition is broad rather than sharp: a pool sampling only the extremes suggests a clean
gap, and filling the middle removes it, so we quote the trend, not a threshold. The ratio
of model error to split, the natural thing to reach for, does \emph{not} serve: over the
same set it runs from $1.00$ to $4.39$ among the successes and from $0$ to $13.1$ among
the failures, and its two largest values belong to failures, because a systematic misfit
harms the discrete selection only through its projection onto the candidate directions, a
projection the ratio does not measure. Quote $R$ and the model error together, with the
basis and count stated. Two crossings are involved, and they are not the same. At fixed
completeness the split rises with bin count and crosses a one per cent floor between
twenty-four and thirty-six bins at seventeen modes per pair, and between seventy-two and
ninety-six at thirty-five: that is detectability, and it is why frequency resolution is
needed at all. The model error falls with bin count and crosses the same floor near one
hundred and forty-four bins along $\rho=1/3$: that is the floor itself, and it is what
fixes the requirement below. The crossing is mesh-sensitive at the level of one bin step,
so we quote the verified point rather than the crossing: on a converged mesh the model
error first dips below one per cent at one hundred and thirty-two bins, but only barely,
$9.9\e{-3}$ against the floor, whereas at one hundred and forty-four it is
$9.0\e{-3}$, a ten per cent margin, and it is there that the recovery battery was run.
The mesh moves these numbers by a few per cent in either direction, so a crossing quoted
to a single bin step would be spurious precision. For Cu(1,3-bdc), where the bins span
the $38$ to $473$~GHz support of the band-pair gap, one hundred and forty-four bins is
about $3.0$~GHz against $9.1$~GHz at forty-eight.

The width of that window is a requirement, not a convenience, and it is worth saying why,
since covering it is the largest experimental demand the protocol makes. Truncating the
window destroys the extraction, and abruptly. Masking the bins above a cutoff and holding
$\rho=1/3$ against the bins that remain, the split response at cutoffs of $141$, $186$
and $232$~GHz is $10^{-14}$ or smaller, exactly degenerate with the nuisance, and at
$322$~GHz it is $6\e{-8}$, still six orders below its full-window value; recovery is zero
in every case. At $413$~GHz the split returns to $4.4\e{-2}$ with $R=115$, and recovery is
still only three per cent. That last row is the
one to note. A high $R$ with no recovery is not a weak measurement but a biased one: the
truncated fit confidently selects a wrong triple, so a partial window returns a wrong
answer rather than a null result.

The mechanism is a property of the band structure rather than of the signal. The three
band pairs occupy nearly disjoint stretches of the gap coordinate: at $D/J=0.2$ the pair
$(1,2)$ spans $50$ to $385$~GHz, the pair $(2,3)$ spans $95$ to $192$, and the pair
$(1,3)$ spans $243$ to $485$. The widest pair therefore contributes nothing below
$243$~GHz and is the sole contributor above $385$. Cutting the top of the window removes
the only region in which that pair is isolated, its content goes unmeasured, and the
split direction, a difference of contents involving it, becomes degenerate. The
requirement that every band pair be resolved somewhere sets the window, and the highest
pair fixes the ceiling. This is also the sense in which the protocol is a spectroscopy
rather than a noise measurement with a frequency axis attached: it needs the whole
spectrum because each pair is a separate observable.

\subsection{Experimental design under matched flexibility}

Table~\ref{tab:design} re-runs the design comparison at matched nuisance flexibility,
using the largest mode count every design supports. Three conclusions survive. Frequency
resolution alone is insufficient: at a single setting the split is degenerate with the
nuisance to $2\e{-14}$, exactly zero, so this is a statement about degeneracy, not a
small number. At this mode count nine settings match one hundred and forty-four at every
noise level tested, a sixteen-fold reduction in measurement. And temperature is the more
important axis.

\begin{table*}[t]
\caption{Design comparison at matched nuisance flexibility (twelve modes per pair, the
largest every design supports), forty-eight bins, atomic convention, one thousand trials
per level. Recovery of the true triple at the stated per-point additive noise. The script
producing this table is named in Appendix~\ref{app:code}, as are those for every other
computed quantity. The binomial standard error is at most $1.6$ percentage points;
entries at $0$ and $100$ carry a rule-of-three bound of $0.3\%$ rather than an exact
value.}
\label{tab:design}
\begin{ruledtabular}
\begin{tabular}{lcccccc}
design & points & split & $1\%$ & $5\%$ & $10\%$ & $20\%$\\
\colrule
single $(T,h)$ & $48$   & $2\e{-14}$   & $0$   & $0$   & $0$   & $0$\\
$3T\times1h$   & $144$  & $6.6\e{-2}$  & $100$ & $100$ & $99.6$ & $91.2$\\
$1T\times3h$   & $144$  & $1.1\e{-2}$  & $99.9$ & $71.0$ & $45.3$ & $25.0$\\
$3T\times3h$   & $432$  & $8.0\e{-2}$  & $100$ & $100$ & $100$ & $100$\\
full $12\times12$ & $6912$ & $7.6\e{-2}$ & $100$ & $100$ & $100$ & $100$\\
\end{tabular}
\end{ruledtabular}
\end{table*}

At forty-eight bins these results and the identifiability argument of
Sec.~\ref{sec:partii} sit at different nuisance flexibilities, twelve to seventeen modes
per pair against thirty-five, and the two do not coincide: at thirty-five modes the
nine-setting recovery collapses, while at seventeen the model error is four per cent.
That apparent tension is an artifact of holding the resolution fixed, and
Fig.~\ref{fig:gap} shows why. What controls the split response is not the mode count but
its ratio to the resolved dimension, $\rho=M_{2}/N_{\rm bins}$. Below
$\rho\simeq0.45$ the split sits at its natural size of a few times $10^{-2}$; above it
the nuisance spans the resolved space in the gap coordinate and the split collapses by
two orders, taking the recovery with it. Growing the nuisance at fixed resolution thus
buys model error by spending signal, which is the trade the previous subsection warned
about, in quantitative form.

Along a line of fixed $\rho=1/3$ the trade disappears. The model error falls roughly as
$N_{\rm bins}^{-1.4}$ while the split stays put near $6\e{-2}$, so the scissors open by
raising resolution rather than mode count, and they open without limit. A one per cent
model error is reached at $N_{\rm bins}=144$ with $M_{2}=48$, and at that point every
claim in this section holds at once. At five per cent additive noise the recovery of
$(1,0,-1)$ is one thousand of one thousand trials for exact calibration and for
$J\pm1\%$, $J+5\%$ and $D+5\%$ alike, on the full grid and on nine settings, while the
misfit still announces itself, the residual standing at six times the exact-model floor
for a one per cent error in $J$ and thirteen times for five per cent. At forty-eight bins
the same five per cent error recovered nothing.

The rule $\rho\le1/3$ does two jobs, the second the less obvious. One hundred and twenty
bins is not enough, and the reason is not the floor alone: at $\rho=1/3$ the model error
there is $1.2\e{-2}$, and the mode count that does reach one per cent at that resolution,
$M_{2}=44$, sits on the knife edge of Fig.~\ref{fig:recovery}, which has not disappeared
at high resolution but has become a property of $\rho$. At every resolution tested the
miscalibrated recovery collapses for $\rho$ between about $0.36$ and $0.39$, well inside
the $\rho\simeq0.45$ collapse of the exact-model split, and it does so non-monotonically:
at one hundred and twenty bins $J+5\%$ recovers fully at $M_{2}=42$, not at all at $44$,
fully again at $46$ and not at all at $48$. Holding $\rho$ at a third keeps the protocol
clear of that edge as well as of the signal collapse, which is why we state it as the
rule rather than stating a mode count.

The closure is conditional on the design carrying a temperature axis: 
at one hundred and twenty bins with forty-two modes per pair, on the $N=90$ mesh, the
three-temperature designs and the full grid reach $R=12.5$, $22.4$ and $113$ and recover
completely,
while the field-only $1T\times3h$ design reaches $R=0.2$ and recovers
seven per cent, so no resolution rescues it. The remaining requirement is experimental
rather than theoretical. The bins span the full support of the band-pair gap, which for
Cu(1,3-bdc) at $D/J=0.15$ is $38$ to $473$~GHz, so one hundred and forty-four bins is
about $3.0$~GHz and forty-eight bins about $9.1$~GHz; the $38$--$147$~GHz figure quoted
elsewhere is where the main weight sits, not the range the bins cover. We state the
design and calibration results at twelve modes because that is where they were computed
and where the comparison across designs is fair, and we state the resolution requirement
as the condition under which they and the identifiability argument coexist.

Recovery also improves with the number of data points as well as with the signal, since
per-point noise averages down as the square root of the count: for forty-eight bins
across one hundred and forty-four settings there are $6912$ points and $\sqrt{6912}=83$,
so twenty per cent per point is a quarter of a per cent in aggregate. Design and noise
level must therefore be quoted together.

\subsection{Calibration: dispersion, thermometry, and field}

The protocol assumes the band structure is known. Table~\ref{tab:calib} shows what
happens when it is not. Dispersion is by far the tightest requirement: $J$ must be right
to about $0.1$ per cent and $D$ to about $0.5$ per cent, and a five per cent error in $J$
produces a confidently wrong triple. Field is looser, at roughly one per cent, and
temperature is benign, five per cent leaving recovery unchanged because the
induced drift is nearly along the gauge direction.

\begin{table*}[t]
\caption{Calibration sensitivity, full grid, complete nuisance (thirty-five modes per
pair), five per cent per-point additive noise, one thousand trials, atomic convention.
Data generated at the true parameters, model built at the perturbed ones. The residual is
the full-model model error, $4.3\e{-5}$ for exact calibration. Columns four and five give
the residual and recovery for the \emph{negative} error of the same size; the binomial
standard error is at most $1.6$ percentage points throughout. The dispersion rows are
strongly asymmetric and the thermal rows are not, for the reason given in the text.}
\label{tab:calib}
\begin{ruledtabular}
\begin{tabular}{lccccc}
miscalibration & residual & zero-noise fit & recovery & residual$^{-}$ & recovery$^{-}$\\
\colrule
exact     & $4.3\e{-5}$ & $(1.00,0.08,-1.08)$  & $58.7$ & --- & ---\\
$J\pm1\%$ & $1.6\e{-3}$ & $(1.05,0.21,-1.26)$  & $37.4$ & $1.5\e{-3}$ & $55.7$\\
$J\pm5\%$ & $8.2\e{-3}$ & $(1.25,-0.11,-1.14)$ & $0.9$  & $9.3\e{-3}$ & $47.2$\\
$D\pm1\%$ & $1.2\e{-4}$ & $(1.00,-0.02,-0.99)$ & $19.4$ & $2.9\e{-4}$ & $54.1$\\
$D\pm5\%$ & $2.2\e{-2}$ & $(0.98,7.19,-8.17)$  & $4.2$  & $4.5\e{-4}$ & $75.8$\\
$T\pm1\%$ & $2.4\e{-3}$ & $(1.05,-0.35,-0.71)$ & $59.3$ & $2.5\e{-3}$ & $57.4$\\
$T\pm5\%$ & $1.2\e{-2}$ & $(1.26,-1.77,0.52)$  & $62.5$ & $1.3\e{-2}$ & $51.0$\\
$h\pm1\%$ & $2.3\e{-3}$ & $(0.92,0.62,-1.53)$  & $59.5$ & $2.4\e{-3}$ & $56.1$\\
$h\pm5\%$ & $1.2\e{-2}$ & $(0.57,2.81,-3.38)$  & $42.1$ & $1.2\e{-2}$ & $35.4$\\
\end{tabular}
\end{ruledtabular}
\end{table*}

The tolerance is one-sided, and only for the dispersion. Comparing the two signs at fixed
magnitude, $J+1\%$ recovers $37.4\%$ against $55.7\%$ for $J-1\%$, and $D+5\%$ recovers
$4.2\%$ against $75.8\%$ for $D-5\%$, while the thermal rows are symmetric to within
their error bars, $59.3$ against $57.4$ at $T\pm1\%$ and $59.5$ against $56.1$ at
$h\pm1\%$. The reason is that $J$ and $D$ enter the gap function and so move the
analyst's bin edges and the turn-on at $2\sqrt3\,D$ relative to the truth, whereas $T$
and $h$ enter only the thermal weights and leave the binning alone. A dispersion
tolerance quoted as $\pm x\%$ therefore understates one direction and overstates the
other, and we quote both.

Two facts turn this from a defect into a specification. The misfit is self-announcing: a
wrong calibration lifts the residual between one and a half and three orders above the
exact-model floor, so it cannot hide, though the margin is smallest exactly where the
error is smallest, a tenth of a per cent in $J$ lifting it by well under an order. And
co-fitting largely cures it: profiling the calibration parameters over a grid jointly
with the candidate triple, on the grid
$J'\times D'=\{0.95,1.00,1.05\}\times\{0.19,0.20,0.21\}$ with the truth on a node,
recovers the correct triple in $89.4(1.4)\%$ of trials under relative per-point noise and
identifies the correct calibration node in $96.6(0.8)\%$. Under additive noise at the
same level the triple is recovered in $61.8(2.2)\%$ and the node in $57.0(2.2)\%$, so the
remedy is substantially weaker there and the two should not be averaged. The same data
over-determine the calibration precisely because the order-one sensitivity that makes the
protocol fragile makes the calibration easy.

We therefore specify the protocol as follows: \emph{the band structure, temperature and
field are co-estimated from the measured spectrum itself}, with independently determined
values as the seed and the residual as the goodness-of-calibration diagnostic; the Chern
triple is read off at the calibrated node. The effective sub-per-cent calibration is
delivered by the fit rather than demanded of the input data.

\section{Application to Cu(1,3-bdc)}
\label{sec:material}

We collect the specification in one place. The material is the kagome ferromagnet
Cu(1,3-bdc), for which the magnon bands and their Chern numbers are established and the
thermal Hall effect has been measured~\cite{Hirschberger,Chisnell}. Neutron parameters
give $J=0.6$~meV, $S=1/2$ and $D/J\approx0.15$; the Land\'e factor inferred from the
exponential suppression of the thermal Hall signal is $g\approx1.6$.

The magnon energy unit is $JS=0.3$~meV, and every dimensionless quantity in the tables
above is quoted in it; the tables themselves are computed at the model value $D/J=0.2$,
the nearest round value to the material's $0.15$, kept for continuity with the convention
checks of Appendix~\ref{app:sign}. At the material's own $D/J=0.15$ a Zeeman gap of
$0.3\,JS$ requires about $1$~T, inside the range already swept. The spectrum turns on at
the minimal direct gap $2\sqrt3\,D=0.52\,JS$, which is $38$~GHz, with the main weight
near $2.0\,JS$, or $147$~GHz; ninety-six bins is about $4.5$~GHz resolution and
forty-eight bins about $9.1$~GHz, the bins spanning the full $38$--$473$~GHz support of
the band-pair gap rather than the interval just quoted, which is where the weight sits.
The design is nine settings, three temperatures by three fields, with a spectrum at each.

Two of these numbers are material-specific and the rest are not, which gives the
experiment a design freedom worth stating. The protocol is scale-free in $JS$: the gap
support, the temperature range and the Zeeman fields are all fixed multiples of it, so
they move together, and only their common scale is a property of the compound.
Table~\ref{tab:scaling} follows that family. Cu(1,3-bdc) sits at the top row, where the
frequency reach is the demanding entry and the cryogenics undemanding; a compound with
half the exchange halves the reach and the fields at the cost of a lower base temperature,
and one with a quarter brings the ceiling to $121$~GHz, within reach of a single
millimetre-wave front end, at $0.44$~K and sub-gigahertz bins. The reach and the
cryogenics therefore trade against each other along a one-parameter family, and a
proposal should choose where on it to sit rather than treat the numbers of the top row as
the specification. For Cu(1,3-bdc) itself the window spans the range covered by
heterodyne receivers in radio astronomy, where cross-correlation is the native
measurement, so the requirement is that the spectrum be assembled from sub-bands under a
common calibration rather than acquired in one instrument.

\begin{table*}[t]
\caption{The protocol is scale-free in the magnon energy unit $JS$, so the frequency
window, the temperatures and the fields are fixed multiples of it and move together. The
first row is Cu(1,3-bdc). Frequencies are the full band-pair gap support at $D/J=0.2$;
the bin width is at one hundred and forty-four bins; temperatures span $0.5$ to $2\,JS$
and fields $0.3$ to $2\,JS$ at $g=1.6$.}
\label{tab:scaling}
\begin{ruledtabular}
\begin{tabular}{ccccc}
$JS$ (meV) & window (GHz) & bin width (GHz) & $T$ (K) & $B$ (T)\\
\colrule
$0.30$  & $50$--$485$ & $3.02$ & $1.7$--$7.0$  & $1.0$--$6.5$\\
$0.20$  & $34$--$324$ & $2.01$ & $1.2$--$4.6$  & $0.7$--$4.3$\\
$0.15$  & $25$--$243$ & $1.51$ & $0.9$--$3.5$  & $0.5$--$3.2$\\
$0.10$  & $17$--$162$ & $1.01$ & $0.6$--$2.3$  & $0.3$--$2.2$\\
$0.075$ & $13$--$121$ & $0.76$ & $0.4$--$1.7$  & $0.2$--$1.6$\\
\end{tabular}
\end{ruledtabular}
\end{table*}

The observable is the antisymmetric part of the current cross-spectrum between orthogonal
components, which suppresses uncorrelated detector noise; that suppression averages down
as the square root of the sample count rather than vanishing, so the requirement falls on
integration time,
\begin{equation}
\label{eq:tint}
T_{\rm int}\;\ge\;\frac{S^{\rm det}_xS^{\rm det}_y/(S^{A})^{2}}{\Delta\nu},
\end{equation}
which at a resolution bandwidth of order $1$~GHz gives about a millisecond for a
floor-to-signal ratio of $10^{3}$, $0.1$~s for $10^{4}$ and $10$~s for $10^{5}$. Integration
time is not the binding constraint; bandwidth and interface transparency at these
frequencies are. The pickup geometry, finally, is unconstrained by the subtraction: by
Lemma~\ref{lem:chiral} no magnetization auto-correlation enters $S^{A}$ whatever the detection
realizes, a complete section, an interfacial pickup or a probe at standoff, the
equilibrium correlator on this lattice satisfying the evenness hypothesis of
Sec.~\ref{sec:magnetization}; only the symmetric floors in the ratio $r$ and in
Eq.~\eqref{eq:tint} acquire the curl admixture at standoff, a budget item rather than a
correctness one.

\section{Three dimensions}
\label{sec:3d}

Everything above is two-dimensional, where the co-exact sector of the curvature two-form
is empty and only two of the three Hodge sectors are exercised. In three dimensions all
three appear, the co-exact sector carrying the monopole sources. A three-dimensional
bosonic Weyl prototype confirms the structure: the slice Chern numbers step across the
nodes, the monopole charges are $\pm1$, the decomposition closes exactly, and all three
sectors contribute to the thermal Hall response. In our prototype the co-exact
contribution peaks at about twelve per cent of the harmonic, so it is present but not
dominant, and we state that rather than gloss it.

We do not develop the three-dimensional case here, for a reason of positioning rather
than difficulty: Weyl magnons in pyrochlore ferromagnets and their chiral
anomaly~\cite{MookWeyl,SuWang} are established, so a three-dimensional mean-side
calculation would be a second benchmark rather than a new opening. The noise question in
three dimensions remains open and is the natural continuation of this work.

\section{Discussion}
\label{sec:discussion}

The result of this paper is a negative statement turned positive. The electronic theorem
that topology is silent in transport noise is true, and its proof short, but it rests on a
hypothesis nobody states, namely that the topological sector is a single number. That
hypothesis is a property of the Fermi surface, not of the Berry curvature. Remove the
Fermi surface, as bosons do, and the hypothesis fails; what replaces it is the
occupation-weighted dispersion of Chern numbers, a genuinely different observable and one
we have not found in the literature.

Three things were needed to make this measurable rather than merely true. The
energy-magnetization subtraction had to be settled at the level of fluctuations, and it
dissolved: the subtracted term is a curl, it delivers nothing through a complete
cross-section realization by realization, and the antisymmetric cross-spectrum cannot
receive the magnetization auto-correlation at any wavevector.
The driven protocol had to be
abandoned (Appendix~\ref{app:driven}), because no achievable magnetic gradient drives the
signature to within twelve orders of the equilibrium floor, and the thermal case, estimated
there rather than derived, lands on the same ladder.
And the
equilibrium route had to be resolved in frequency, because for dispersive bands the
integrated form cannot in practice separate the sectors, for a rank reason we exhibit
rather than assert,
and for a second reason unconnected with rank, since the frequency integral also
pools the intrinsic channel against the extrinsic one (Sec.~\ref{sec:parti}). That last
point is where the paper meets its nearest neighbor. That current noise carries band
geometry, and that frequency resolution resolves the gap coordinate, are both established
at zero temperature for band insulators~\cite{NeupertChamonMudry}; what was left open
there was the multiband case, where energetics and geometry combine. This paper answers
that question in the setting where it has an answer, and the finite-temperature occupation
weighting that makes the observable a dispersion rather than a single invariant is the
same feature that makes the multiband case tractable.

We close with the limitations. The result is established for collinear ferromagnetic
magnons, where number conservation holds; Bogoliubov systems need the symplectic
construction and are untreated. The protocol's calibration requirement is severe and is
met by co-estimation rather than external measurement, a design choice that should be
tested against real data rather than simulated. The three-dimensional case, where the
co-exact sector is non-empty, is open. And the fluctuation theory of statistical drives,
which would underwrite a real temperature gradient at the level of noise, does not yet
exist; it is delimited in Appendix~\ref{app:driven}, together with the finite-$q$
statistical-magnetization fluctuation and the $O(q^{2})$ counting of the
transport-magnetization cross term, and belongs to the heat-channel continuation. A
second limitation of scope belongs here too. The vertices throughout are the intrinsic
ones, and disorder enters only through the relaxation rates of Sec.~\ref{sec:collision};
extrinsic side-jump and skew-scattering contributions, which for bosonic thermal Hall
transport can match the intrinsic ones and can survive where they
vanish~\cite{MangeolleKnolle}, are not treated. Two features of that treatment bound its
reach here. The disorder of Ref.~\cite{MangeolleKnolle} is correlated in space but static
in time, so the scattering is elastic and none of their impurity classes enters the
energy-relaxation channel, which makes the widths-only assignment of
Sec.~\ref{sec:collision} safe within their treatment; and their construction builds on
Ref.~\cite{MangeolleSavaryBalents} and cites Ref.~\cite{QinNiuShi} for the
energy-magnetization complications, so bound currents are already isolated in their
framework and Theorem~\ref{thm:vanish} adds nothing there. Whether the extrinsic
contributions contaminate the extraction is a question we leave open, with two remarks
that bound it. A side-jump correction modifies the current vertex, hence the observable
being projected, so the decomposition alone does not exclude it from the number-mode
weight. For the frequency-resolved protocol, however, the extraction window lies at
interband frequencies, where quasi-elastic extrinsic contributions enter only through
collisional broadening, suppressed by the ratio of the scattering rate to the minimal
direct gap $2\sqrt3\,D$; and a smooth extrinsic admixture in the gap coordinate is
absorbed by the nuisance, biasing the selection only through its projection onto the
candidate directions, which the residual diagnostic of Sec.~\ref{sec:robust} monitors.
Turning the first remark into a theorem, with the small parameter stated, belongs to the
continuation.

The second remark can be measured rather than asserted, and it is more robust than we
expected. Broadening the data in the gap coordinate with a Lorentzian kernel, normalized
column-wise and with absorbing edges so that weight carried below the turn-on is lost,
and fitting the unbroadened model to it, the recovery of $(1,0,-1)$ at the closure point
and five per cent noise is unaffected out to $\Gamma=0.35\,JS$, half the turn-on at the
model coupling $D/J=0.2$ at which the battery is run.
The continuous fit drifts along the gauge-adjacent direction $(0,-1,1)$, reaching
$(1.00,-0.33,-0.67)$ at that width and still rounding correctly, and the misfit announces
itself as it should, the residual rising from $9.0\e{-3}$ to $3.9\e{-2}$ and then
saturating. The discreteness of the candidate lattice is doing here what it does for the
five per cent dispersion tolerance of Sec.~\ref{sec:robust}: a bias that would ruin a
continuous estimate leaves the integer selection intact.

One consequence deserves statement as a check declared in advance, because it costs the
experiment nothing. The intrinsic theory predicts $S^{A}$ identically zero below the
turn-on, no band pair having a gap smaller than $2\sqrt3\,D$. Any antisymmetric weight
measured on $[0,2\sqrt3\,D)$ is therefore not signal but a direct reading of the
extrinsic admixture together with whatever asymmetry the apparatus contributes, taken
from the same run and the same calibration as the data. The window has a null channel
attached to it, and an experiment that reports the spectrum below the turn-on reports its
own systematic error along with its result.

\section*{Data availability}

This is a theoretical and computational study; no experimental data were generated. Every
number, table and figure in the paper is produced by the deposited code described in
Appendix~\ref{app:code}, openly available in the Zenodo repository at
\mbox{DOI:10.5281/zenodo.21945238}. The scripts require only \textsc{Python}~3 with
\textsc{NumPy} and \textsc{Matplotlib}; two of them additionally use \textsc{SymPy} or
\textsc{mpmath}, and none reads an external data file, so the results are regenerated
from the model rather than from stored output.

\begin{acknowledgments}
The numerical results were obtained with two independent implementations, written from
the specification rather than from shared code; both are deposited. We thank the authors
of the works cited in Sec.~\ref{sec:claims} for the frameworks this paper applies.
\end{acknowledgments}

\appendix

\section{Conventions and the sign anchor}
\label{app:sign}

Two sign conventions circulate, differing by an overall factor, so we fix ours against an
anchor that uses neither eigenvectors nor a Berry connection. For a two-band
$H=\mathbf d\cdot\boldsymbol\sigma$ the degree integral
$\frac{1}{4\pi}\int\hat{\mathbf d}\cdot(\partial_x\hat{\mathbf d}\times\partial_y\hat{\mathbf d})$
is pure geometry of the $\hat{\mathbf d}$ map and admits no gauge or orientation
ambiguity. In Berry's convention, $A_\mu=i\langle u|\partial_\mu u\rangle$ and
$\Om=\partial_xA_y-\partial_yA_x$, one has
$\Om_{\rm lower}=+\frac12\hat{\mathbf d}\cdot(\partial_x\hat{\mathbf d}\times\partial_y\hat{\mathbf d})$
exactly, and the consistent Kubo form is
\begin{equation}
\label{eq:kubo}
\Om_n=-2\,\mathrm{Im}\sum_{n'\neq n}
\frac{\langle n|\partial_xH|n'\rangle\langle n'|\partial_yH|n\rangle}
{(\varepsilon_n-\varepsilon_{n'})^{2}} .
\end{equation}

Two traps follow. First, the verbatim Fukui--Hatsugai--Suzuki link product equals
$\exp(\oint\langle u|\dd u\rangle)=\exp(-i\oint A)$, so its angle is minus the Berry flux
and a by-the-book implementation returns $-C$; the fix is to negate the plaquette angle.
Second, specific to the kagome model, the symmetric $\cos(\mathbf k\cdot\mathbf d_{ij})$
form of $H$ is only antiperiodic under reciprocal translations, since
$\mathbf b_1\cdot\mathbf d_{AB}=\pi$, so a plaquette sum over a single zone in that gauge
joins eigenvectors of different Hamiltonians and is silently wrong. It requires the
canonical Bloch gauge, $H_{AB}=-(1+iD\nu)(1+e^{-i\mathbf k\cdot\mathbf a_1})$ and cyclic,
which is exactly periodic. The Kubo route is pointwise and immune to this second trap,
though not, as the next paragraph shows, to the third. Both went live during the
independent reimplementation, and both were caught by this anchor.

The third trap is that the two Bloch conventions are not interchangeable for anything
pointwise. They differ by a $\mathbf k$-dependent unitary diagonal in sublattice index,
which is not a $U(1)$ gauge transformation of the eigenvectors, so it redistributes $\Om$
and $m$ across the zone. Only in the atomic convention, where the intracell positions sit
in the phases, is $\partial H/\partial\mathbf k$ the true velocity operator; the
diagnostic is parity, which on this inversion-symmetric lattice holds to machine zero in
the atomic convention and fails at order unity in the periodic one. The rule is
therefore: $C$ may be computed in either convention, but every pointwise, weighted or
nonlinear quantity, and every parity argument, including the vanishing of the
transport-magnetization cross term at $q=0$ and the $O(q^{2})$ counting that follows
it, is computed in the atomic convention, while the link product is computed in the
periodic one. The second and third traps thus pull in opposite directions, and the two
computations belong in different conventions. One caution for anyone reusing these
numbers: the agreement of $\sum\!\int\varepsilon\,\Om$ between the two conventions on
this model is not a general fact but an accident of inversion symmetry, which makes the
difference parity odd; breaking inversion by scaling the up-triangle bonds separates the
two values by several per cent while $C$ still agrees.

With $\nu_{AB}=\nu_{BC}=\nu_{CA}=+1$ and $D>0$ the anchored convention gives
$C=(+1,0,-1)$, lowest band first, from both routes. 
The chain ties to experiment: under
this convention $\sum\!\int c_2\Om>0$ across the verified range $0.05\le T\le100\,JS$,
sampled densely in $T$ at meshes $N=90$ and $120$, so $\kxy<0$ throughout it, while
$\sum\!\int\varepsilon\,\Om\,\dd^{2}k=-19.9$ at the model coupling $D/J=0.2$, in $JS$
units and the plain $\dd^{2}k$ measure, without the $1/2\pi$ of Eq.~\eqref{eq:harmonic},
so the coefficient of the $1/T$ tail of $\sum\!\int c_2\Om$ is positive and $\kxy$
saturates at the negative constant it fixes rather than changing sign; that limit is the
high-temperature figure of merit of Ref.~\cite{MookHenkMertig2014}. In between, the
weight itself is not monotone: $\sum\!\int c_2\Om$, and with it $|\kxy|/T$ by
Eq.~\eqref{eq:kxy}, peaks at $T=1.64\,JS$ at the model coupling and at $1.66\,JS$ at the
material's $D/J=0.15$, while $T\sum\!\int c_2\Om$, hence $|\kxy|$ itself, rises
monotonically to its saturation value.

Two further anchors serve the self-rotation vertex of Sec.~\ref{sec:magnetization}. In
the eigenbasis,
\begin{equation}
\label{eq:manchor}
m_n(\mathbf k)\;=\;+\frac{1}{\hbar}\,\mathrm{Im}\sum_{n'\neq n}
\frac{\langle n|\partial_xH|n'\rangle\langle n'|\partial_yH|n\rangle}
{\varepsilon_n-\varepsilon_{n'}}\,,
\end{equation}
with the plus sign and the single energy denominator against the minus two and squared
denominator of Eq.~\eqref{eq:kubo}; and for two bands
$m_{\rm lower}=\tfrac{1}{2\hbar}(\varepsilon_2-\varepsilon_1)\,\Om_{\rm lower}$ exactly.
The second is the natural arbiter row for the new vertex; a sign slip of the kind it
catches arose in one implementation and was caught by it.

\section{The driven protocol does not work}
\label{app:driven}

This appendix quantifies the driven route to the observable and rules it out; the
arithmetic is stated plainly rather than left implicit.

The drive must be a force, and for magnons there is a natural one: a static field
gradient. Each magnon raises the Zeeman energy by $g\mu_BB$, so a spatially varying field
is a genuine potential landscape $U(\mathbf r)=g\mu_BB(\mathbf r)$ per magnon, exerting
$\mathbf F=-g\mu_B\nabla B$, the same for every band and every $\mathbf k$. Stronger
still: the Zeeman potential is proportional to the identity in the band basis, so the
gradient rigidly shifts all bands and leaves the eigenvectors, hence the curvature and
the self-rotation moment, exactly unchanged. The drive is geometry-neutral and cannot
dress the vertices whose statistics we compute. The same apparatus supplies both knobs,
the uniform part $B_0$ setting the Zeeman gap and the gradient doing the driving, and the
mean transverse response to $\mathbf F$ is the transport coefficient with the plain
occupation weight, $\sigma^{N}_{xy}=\hbar^{-1}\sum_\lambda\int\Om_\lambda n_\lambda$, with
no $c$-function replacement, since the $c$-type weights arise only for statistical drives
or for energy-current vertices. Theorem~\ref{thm:vanish} covers the noise.

The excess transverse noise relative to the equilibrium floor is $\eta^{2}$ with
\begin{equation}
\label{eq:eta}
\eta \;=\; \frac{g\mu_B\,|\nabla B|\,a_0}{JS},
\end{equation}
the drive energy across a lattice constant measured against the exchange scale. For
Cu(1,3-bdc), with $g=1.6$, $J=0.6$~meV, $S=1/2$ and lattice constant $a_0\approx1$~nm, so
that the magnon energy unit is $JS=0.3$~meV and $g\mu_B=0.093$~meV/T, and since
$g\mu_B/\kB=1.07$~K/T a tesla per meter drives almost exactly as hard as a kelvin per
meter: $\eta^{2}\approx1\times10^{-11}$ at $|\nabla B|=10^{4}$~T/m, the scale of
patterned pole pieces or on-chip micromagnets over micron distances, and $1\times10^{-7}$
at $10^{6}$~T/m, reached only within roughly a hundred nanometers of a force-microscopy
tip, where the driven volume shrinks with the gradient. The electronic counterpart
reaches about $10^{-5}$ in a semimetal at $10^{6}$~V/m, so the magnon version is roughly
two orders harder at matched idealization. The topological part is the few-per-cent step
within the excess, so the signature sits near $2\times10^{-13}$ of the floor at
achievable gradients and near $2\times10^{-9}$ at tip gradients.

A real temperature gradient is a different and harder problem, which we delimit rather
than model. A statistical drive exerts no force, the anomalous vertex is absent, and the
mean intrinsic response arises through the subtraction machinery itself, with the $c_1$
weight of Zyuzin and Kovalev~\cite{ZyuzinKovalev} in the magnon spin channel; the
fluctuation-level form of that machinery does not exist. An estimate treating the thermal
drive as an effective force of order $\kB\nabla T$ lands on the ladder above through the
numerical equivalence just noted, but it is a surrogate rather than a derivation. Nothing
in this paper rests on it, and the statistical-drive fluctuation theory is assigned, with
the heat-channel questions, to the continuation on the framework of Ref.~\cite{GGK}.

\section{Deposited code}
\label{app:code}

The scripts accompanying this paper reproduce every number quoted.
\texttt{verify\_symbolic.py} is the symbolic suite for Theorem~\ref{thm:vanish},
Lemma~\ref{lem:chiral} and the local-equilibrium identity~\eqref{eq:LEidentity};
\texttt{m\_vertex\_kagome.py} evaluates the self-rotation vertex and its spectral shares;
\texttt{sign\_sweep\_harness.py} asserts the convention chain of Appendix~\ref{app:sign}
row by row and exhibits all three traps; \texttt{make\_figures.py} produces
Figs.~\ref{fig:shares}--\ref{fig:recovery} and \texttt{make\_fig5.py} produces
Fig.~\ref{fig:gap}, both from the same modules. The design and calibration studies exist
in two independent versions, both deposited under distinguishing names:
\texttt{design\_A.py} and \texttt{calibration\_A.py} (Tables~\ref{tab:design}
and~\ref{tab:calib}, and the temperature and field rows, computed only by this version)
against \texttt{design\_B.py} and \texttt{calibration\_B.py} (the completeness-matched
design section, the tolerance scan, the relative-noise column and the co-fit). All depend
on the two modules of the independent reimplementation, \texttt{kagome.py} and
\texttt{protocol.py}, together with \texttt{run\_suite.py}. The scripts of the final
revision round sit beside them: \texttt{raised\_calib.py} regenerates
Table~\ref{tab:calib}, \texttt{mode\_count\_gap.py} the closure analysis behind
Fig.~\ref{fig:gap}, \texttt{criterion\_test.py} the comparison of the two selection
criteria across twenty configurations and, in a second block, across the transition band
the first deliberately straddles, and \texttt{closure\_final.py} the closure battery at
$(144,48)$ on the converged $N=120$ mesh, together with \texttt{truncation.py}, the
window-truncation study of Sec.~\ref{sec:robust}. Of these, \texttt{mode\_count\_gap.py}
and \texttt{criterion\_test.py} fix the $k$-mesh at $N=90$; at one hundred and twenty
bins and beyond the mesh must be raised further, as recorded in the convergence section
of \texttt{mode\_count\_gap.py}. Both implementations of the protocol are deposited,
since they differ in nuisance basis, reference density and noise convention, and hence in
every basis-dependent number; agreement on the convention-free invariants across two
independent constructions is the evidence we rely on. One difference is deliberate and
should not be reconciled: the two place the Zeeman gap differently, one inside $H$ and
one in the occupation factor, which is a rigid spectral shift and leaves the curvature,
the self-rotation moment and every covariance unchanged.

\end{document}